\documentclass[conference]{IEEEtran}
\usepackage[ruled,vlined]{algorithm2e}
\usepackage{cite}
\usepackage{graphicx}
\usepackage{psfrag}
\usepackage{subfigure}
\usepackage{url}
\usepackage{amsmath}
\usepackage{array}
\usepackage{amssymb}
\usepackage{amsfonts}
\usepackage{graphicx}
\usepackage{epstopdf}

\newtheorem{lemma}{Lemma}

\newtheorem{definition}{Definition}

\newtheorem{example}{Example}

\title{Wireless Network-Coded Three-Way Relaying Using Latin Cubes}
\begin{document}

\author{
\authorblockN{Srishti Shukla \textsuperscript{$\dagger$}, Vijayvaradharaj T Muralidharan\textsuperscript{$\#$}  and B. Sundar Rajan\textsuperscript{$\dagger$}}\\
\authorblockA{Email: {$\lbrace$srishti, tmvijay, bsrajan$\rbrace$} @ece.iisc.ernet.in\\
\textsuperscript{$\dagger$}IISc Mathematics Initiative (IMI), Dept. of Mathematics and Dept. of Electrical Comm. Engg., IISc, Bangalore\\
\textsuperscript{$\#$} Dept. of Electrical Comm. Engg., IISc, Bangalore
}
}

\maketitle
\pagestyle{plain}	
%%%%%%%%
\begin{abstract}
The design of modulation schemes for the physical layer network-coded three-way wireless relaying scenario is considered. The protocol employs two phases: Multiple Access (MA) phase and Broadcast (BC) phase with each phase utilizing one channel use. For the two-way relaying scenario, it was observed by Koike-Akino et al. \cite{KPT}, that adaptively changing the network coding map used at the relay according to the channel conditions greatly reduces the impact of multiple access interference which occurs at the relay during the MA phase and all these network coding maps should satisfy a requirement called \textit{exclusive law}. This paper does the equivalent for the three-way relaying scenario. We show that when the three users transmit points from the same 4-PSK constellation, every such network coding map that satisfies the exclusive law can be represented by a Latin Cube of Second Order. The network code map used by the relay for the BC phase is explicitly obtained and is aimed at reducing the effect of interference at the MA stage.
\end{abstract}

\section{Background And Preliminaries}
The concept of physical layer network coding has attracted a lot of attention in recent times. The idea of physical layer network coding for the two-way relay channel was first introduced in [1], where the multiple access interference occurring at the relay was exploited so that the communication between the end nodes can be done using a two stage protocol. Information theoretic studies for the physical layer network coding scenario were reported in [2], [3]. The design principles governing the choice of modulation schemes to be used at the nodes for uncoded transmission were studied in [4]. An extension for the case when the nodes use convolutional codes was done in [5]. A multi-level coding scheme for the two-way relaying was proposed in [6]. \\

We consider the three-way wireless relaying scenario shown in Fig. 1, where three-way data transfer takes place among the nodes A, B and C with the help of the relay R. It is assumed that the three nodes operate in half-duplex mode. The relaying protocol consists of two phases, \textit{multiple access} (MA) phase, consisting of one channel use during which A, B and C transmit to R; and \textit{broadcast} (BC) phase, in which R transmits to A, B and C in a single channel use. Network Coding is employed at R in such a way that A(/B/C) can decode B's and C's(/A's and C's /A's and B's) messages, given that A(/B/C) knows its own message.\\

For a two-way wireless relay channel, it was observed in [4] for 4-PSK, that for uncoded transmission, the network coding map used at the relay needs to be changed adaptively according to the channel fade coefficient, in order to minimize the impact of multiple access interference. In other words, the set of all possible channel realizations is quantized into a finite number of regions, with a specific network coding map giving the best performance in a particular region. It is shown in [7] for every M-PSK constellation used by both the users, that every such network coding map that satisfies the \textit{exclusive law} is representable as a Latin square and this relationship can be used to get the network coding maps satisfying the exclusive law. A Latin Square of order $M$ is defined to be an $M \times M$ array in which each cell contains a symbol from $\mathbb{Z}_{t}=\left\{0,1,...t-1\right\}$ such that each symbol occurs at most once in each row and column [8].\\

\begin{figure}[tp]
\center
\includegraphics[height=40mm]{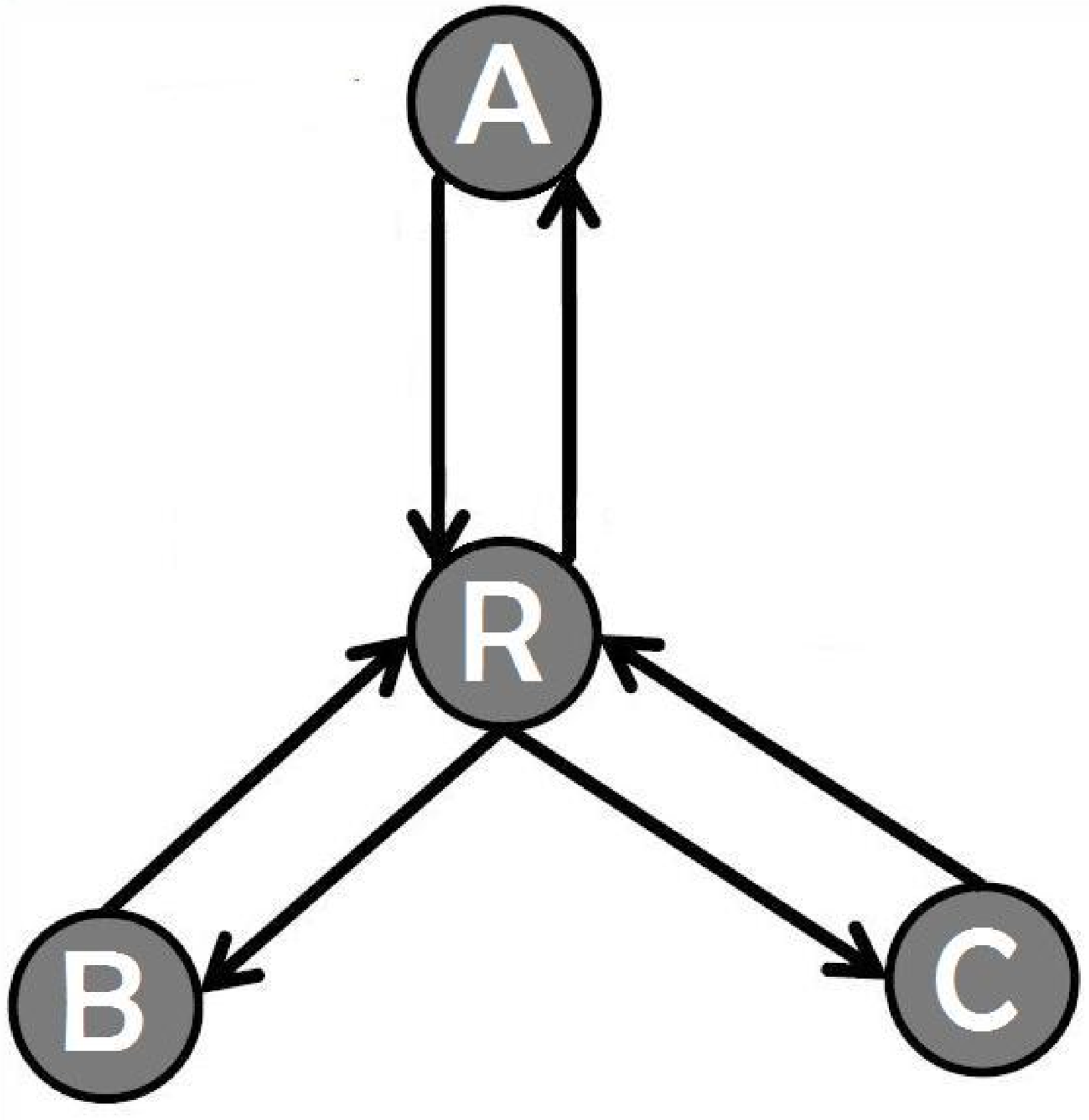}
\caption{A three-way relay channel}
\end{figure}

While most of the research has been done for two-way relay channels, some work has been done for the relay channels with three or more user nodes as well. Liu and Arapostathis [9], proposed a joint network coding and superposition coding for three user relay channels, and claim that the results can be easily extended to information relaying cases with more than three relay nodes. In this scheme, two stage operations are required for encoding and decoding, and four channel uses are required for the information exchange, three channel uses for the MA phase, and one channel use for the BC phase, while for our scheme totally two channel uses suffice. In [9], the first three channel uses are utilized by each user node transmitting its packet to the relay node. Then the relay makes two XOR-ed packets and superimposes them together for broadcast. The relay makes two XOR-ed packets, the packet from the node with the worse channel gain is XOR-ed respectively with the other two packets. Pischella and Ruyet [10] also discuss the three-way wireless relaying scenario, and propose a method of information exchange among the users composed of alternate MA and BC phases. The physical layer network coding strategy for this relaying protocol is a lattice-based coding scheme combined with power control, so that the relay receives an integer linear combination of the symbols transmitted by the user nodes. This scheme also, however, consists of four channel uses. Park and Oh, in [11], propose a network coding scheme for the three-way relay channels and present a `Latin square-like condition' for the three-way network code. They also discuss schemes in order to improve these codes using cell swapping techniques. In this work, though Latin Cubes have been suggested as being equivalent to the map used by the relay, the number of channel uses the scheme uses is five, and the work doesn't deal with the channel gains associated with the channels explicitly. In [12], authors Jeon et al. adopt an `opportunistic scheduling technique' for physical network coding where users in the MA as well as the BC phase are selected on the basis of instantaneous SNR using a channel norm criterion and a minimum distance criterion and plot graphs to show that the proposed scheme outperforms systems without this selection. Their approach, however, utilizes six channel uses. \\

For our physical layer network coding strategy, we use a mathematical structure called a Latin Cube, that has three dimensions out of which one is represented along the rows, one along the columns, and the third dimension is represented along `files'. In our case, we have A's transmitted symbol along the files, B's symbol along the rows, and C's symbol along the columns. For our purposes, we define Latin Cubes as follows:\\

\begin{definition}
A \textit{Latin Cube L of second order of side M} on the symbols from the set $\mathbb{Z}_{t}=\left\{0,1,2,...,t-1\right\}$ is an $M \times M \times M $ array, in which each cell contains one symbol and each symbol occurs at most once in each row, column and file. \\
\end{definition}

The above definition, is given in [13] with $t=M^{2}$.  

\begin{figure}[tp]
\center
\includegraphics[height=40mm]{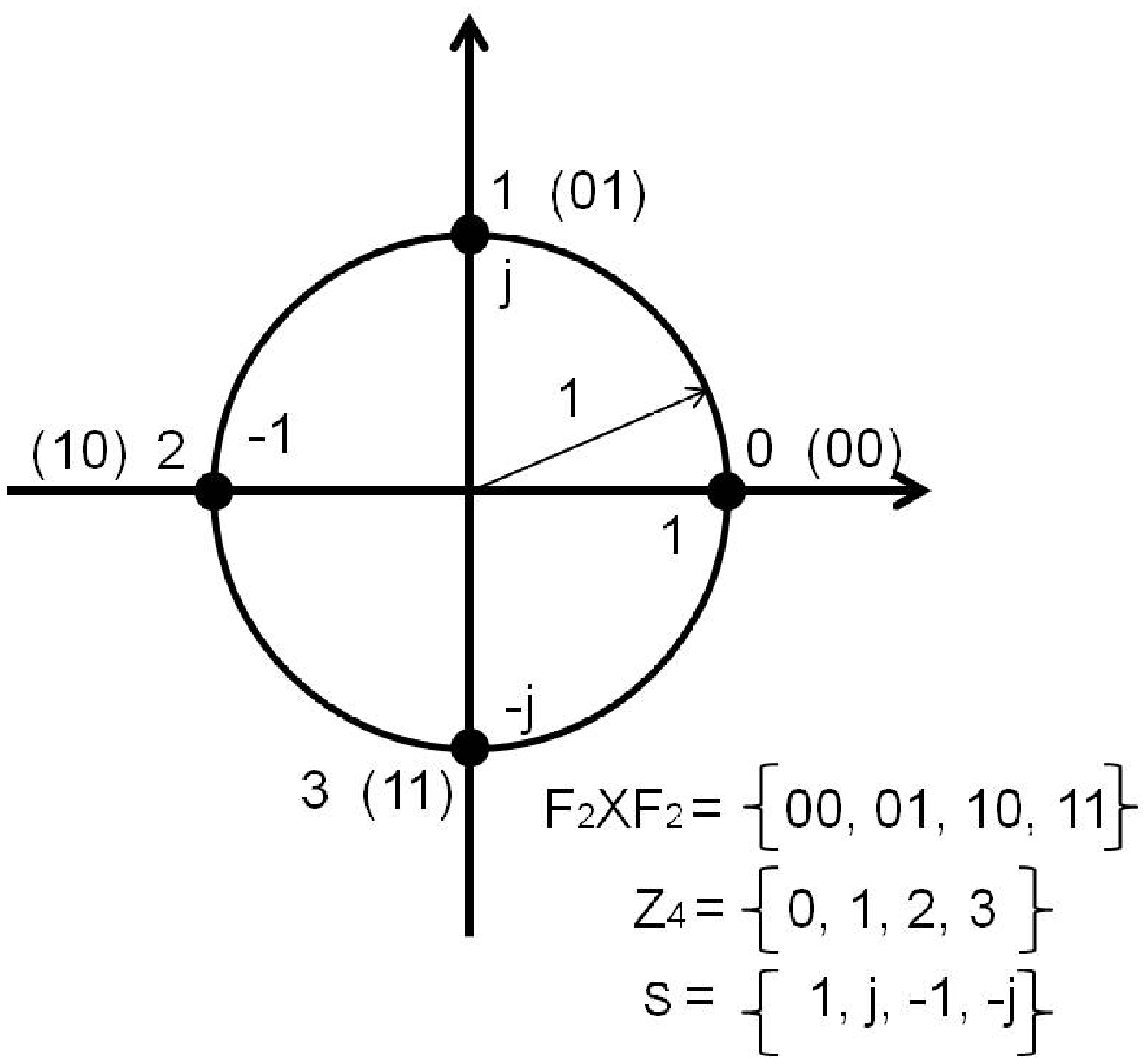}
\caption{4-PSK constellation}
\end{figure}

\subsection{Signal Model}

\noindent \textit{Multiple Access Phase:}\\
\indent Let $\mathcal{S}$ denote the symmetric 4-PSK constellation $\left\{\pm 1,~\pm j\right\}$ as shown in Fig. 2, used at A, B and C. Assume that A(/B/C) wants to send a 2-bit binary tuple to B and C(/A and C/A and B). Let $ \mu : \mathbb{F}^{2}_{2} \rightarrow \mathcal{S} $ denote the mapping from bits to complex symbols used at A, B and C where $\mathbb{F}_{2}=\left\{0,1\right\}$. Let $ x_{A}=\mu\left(s_{A}\right), x_{B}=\mu\left(s_{B}\right), x_{C}=\mu\left(s_{C}\right) \in \mathcal{S}$ denote the complex symbols transmitted by A, B and C respectively, where $s_{A}, s_{B}, s_{C} \in \mathbb{F}^{2}_{2}$. It is assumed that the channel state information is not available at the transmitting nodes A, B and C during the MA phase. The received signal at R in the MA phase is given by 
\begin{equation}
\label{yr}
Y_{R}=H_{A}x_{A}+H_{B}x_{B}+H_{C}x_{C}+Z_{R}
\end{equation}
\begin{figure*}
\footnotesize
\begin{align}
%\hline
\vspace{1 cm}
&
\label{dist}
d_{min}(H_A, H_B, H_C)=\hspace{-0.5 cm}\min_{\substack {{(x_A,x_B,x_C),(x'_A,x'_B,x'_C) \in \mathcal{S}^{3}} \\ {(x_A,x_B,x_C) \neq (x'_A,x'_B,x'_C)}}}\hspace{-0.5 cm}\vert H_A \left(x_A-x'_A\right)+H_B \left(x_B-x'_B\right) + H_C \left(x_C-x'_C\right)\vert \\
\hline
\vspace{1cm}
&
\label{cl1}
d_{min}^{\mathcal{L}_{i},\mathcal{L}_{j}}\left(H_A, H_B, H_C\right)=\hspace{-0.2 cm}\min_{\substack {{(x_A,x_B,x_C) \in \mathcal{L}_{i}},\\ (x'_A,x'_B,x'_C) \in \mathcal{L}_{j}}} \hspace{-0.2 cm}  \left| H_A \left( x_A-x'_A\right)+ H_B \left(x_B-x'_B\right) + H_C \left(x_C-x'_C\right) \right| \\
\hline
\vspace{1cm}
&
\label{cl2}
d_{min} \left(\mathcal{C}^{H_A, H_B, H_C}\right)=\hspace{-0.8 cm}\min_{\substack {{(x_A,x_B,x_C),(x'_A,x'_B,x'_C) \in \mathcal{S}^{3},} \\ {\mathcal{M}^{H_A, H_B, H_C}(x_A,x_B,x_C) \neq \mathcal{M}^{H_A, H_B, H_C}(x'_A,x'_B,x'_C)}}}\hspace{-0.8 cm} \left| H_A \left( x_A-x'_A\right)+ H_B \left(x_B-x'_B\right) + H_C \left(x_C-x'_C\right) \right| .\\
\hline
&
\label{mel1}
\mathcal{M}^{z_{1}, z_{2}}\left(x_{A},x_{B},x_{C}\right) \neq \mathcal{M}^{z_{1}, z_{2}}\left(x_{A},x'_{B},x'_{C}\right), \forall x_{A}, x_{B}, x'_{B}, x_{C}, x'_{C}  \in \mathcal{S}, \ whenever \ \left(x_{B},x_{C}\right) \neq \left(x'_{B},x'_{C}\right)\\
\vspace{1 cm}
&
\label{mel2}
\mathcal{M}^{z_{1}, z_{2}}\left(x_{A},x_{B},x_{C}\right) \neq \mathcal{M}^{z_{1}, z_{2}}\left(x'_{A},x_{B},x'_{C}\right), \forall x_{A}, x'_{A}, x_{B}, x_{C}, x'_{C}  \in \mathcal{S}, \ whenever \ \left(x_{A},x_{C}\right) \neq \left(x'_{A},x'_{C}\right)\\
\vspace{1 cm}
& 
\label{mel3}
\mathcal{M}^{z_{1}, z_{2}}\left(x_{A},x_{B},x_{C}\right) \neq \mathcal{M}^{z_{1}, z_{2}}\left(x'_{A},x'_{B},x_{C}\right), \forall x_{A}, x'_{A}, x_{B}, x'_{B}, x_{C}  \in \mathcal{S}, \ whenever \ \left(x_{A},x_{B}\right) \neq \left(x'_{A},x'_{B}\right)\\
\hline
\vspace{1cm}
&
\label{cl3}
d_{min} \left(\mathcal{C}^{h_A, h_B, h_C} , H_A, H_B, H_C \right)=\hspace{-0.8 cm}\min_{\substack {{(x_A,x_B,x_C),(x'_A,x'_B,x'_C) \in \mathcal{S}^{3},} \\ {\mathcal{M}^{h_A,h_B,h_C}(x_A,x_B,x_C) \neq \mathcal{M}^{h_A, h_B, h_C}(x'_A,x'_B,x'_C)}}}\hspace{-0.8 cm}  \left| H_A \left( x_A-x'_A\right)+ H_B \left(x_B-x'_B\right) + H_C \left(x_C-x'_C\right) \right|.\\
\hline
\nonumber
\vspace{-0.8cm}
\end{align}
\end{figure*}where $H_{A}$, $H_{B}$ and $H_{C}$ are the fading coefficients associated with the A-R, B-R and C-R link respectively. The additive noise $Z_{R}$ is assumed to be $\mathcal{CN}\left(0,\sigma^2 \right)$, where $\mathcal{CN}\left(0,\sigma^2 \right)$ denotes the circularly symmetric complex Gaussian random variable with variance $\sigma^2$. \\
%We assume a block fading scenario, with $z_{1}=\gamma_{1}e^{j\theta_{1}}=H_{B}/H_{A}$ and $z_{2}=\gamma_{2}e^{j\theta_{2}}=H_{C}/H_{A}$ referred to as the \textit{fade states} for the first transmission by A and B at the first channel use, and the second transmission by A and C at the second channel use respectively, where $\gamma_{1}, \gamma_{2} \in \mathbb{R}^+$ and $-\pi \leq \theta_{1}, \theta_{2} \leq \pi$, and for simplicity also denoted by $\left(\gamma_{1}, \theta_{1}\right)$ and $\left(\gamma_{2}, \theta_{2}\right)$. Also, it is assumed that $z_{1}, z_{2}$ are distributed according to a continuous probability distribution. 

Let $ \mathcal{S}_{R} \left( H_A, H_B, H_C \right)$ denote the effective constellations seen at the relay during the MA phase channel use, i.e., 
$$ \mathcal{S}_{R} \left( H_A, H_B, H_C \right) = \left\{H_A x_{i} + H_B x_{j} + H_C x_{k}| x_{i}, x_{j}, x_{k} \in \mathcal{S}\right\}. $$

Let $d_{min}\left(H_A, H_B, H_C\right)$ denote the minimum distance between the points in the constellation $ \mathcal{S}_{R} \left( H_A, H_B, H_C \right) $ as given in (\ref{dist}), where $ \mathcal{S}^{n}=\mathcal{S} \times \mathcal{S} \times .. \times \mathcal{S} \ \tiny{(n \ times)}$. From (\ref{dist}), it is clear that there exists values of $(H_A, H_B, H_C)$, for which $d_{min}\left(H_A, H_B, H_C\right)=0$. Let $\mathcal{H}=\left\{ (H_A, H_B, H_C) \in \mathbb{C}^3 | d_{min}\left(H_A, H_B, H_C\right)=0 \right\}$. The elements of $\mathcal{H}$ are called singular fade states. For singular fade states, $\left|\mathcal{S}_{R} \left( H_A, H_B, H_C \right)\right| < 4^{3}$. \\

\begin{definition}
A fade state $(H_A, H_B, H_C)$ is defined to be a \textit{singular fade state} for the MA phase of three-way relaying, if the cardinality of the signal set $ \mathcal{S}_{R} \left( H_A, H_B, H_C \right)$ is less than $4^{3}$. Let $\mathcal{H}$ denote the set of all singular fade states for the three-way data transfer among A, B and C.\\
\end{definition}

Let $\left(\hat{x}_{A}, \hat{x}_{B}, \hat{x}_{C}\right) \in \mathcal{S}^{3}$ denote the Maximum Likelihood (ML) estimate of $\left(x_{A}, x_{B}, x_{C}\right) $ at R based on the received complex number $Y_{R}$, i.e.,
{
\begin{align}
\footnotesize
\left(\hat{x}_{A}, \hat{x}_{B},\hat{x}_{C}\right)=\arg \min_{\left({x_{A}}, {x_{B}}, {x_{C}}\right) \in \mathcal{S}^{3}}\left\|Y_R - HX\right\|
\end{align}
}
where,
\begin{align}
\vspace{-0.5cm}
\nonumber
\vspace{0.1cm}
& H=\left[H_{A}\  H_{B}\  H_{C}\right]\\
\nonumber
\vspace{0.1cm}
& X=
\left[ {\begin{array}{cc}
\vspace{0.15cm}
x_{A} \\
\vspace{0.15cm}
x_{B} \\
x_{C} \\
\end{array} } \right].
\nonumber
\end{align}

\noindent \textit{Broadcast (BC) Phase:}\\
\indent The received signals at A, B and C during the BC phase are respectively given by,
\begin{equation}
Y_{A}=H_{A}^{'}X_{R}+Z_{A},\ Y_{B}=H_{B}^{'}X_{R}+Z_{B},\ Y_{C}=H_{C}^{'}X_{R}+Z_{C} 
\end{equation}

\noindent where $X_{R}=\mathcal{M}^{H_A, H_B, H_C}\left(\left(\hat{x}_{A},\hat{x}_{B},\hat{x}_{C}\right)\right) \in \mathcal{S}^{'}$ is the complex number transmitted by R. The fading coefficients corresponding to the R-A, R-B and R-C links are given by $H_{A}^{'},$ $H_{B}^{'},$ and $H_{C}^{'}$ respectively and the additive noises $Z_{A},$ $Z_{B}$ and $Z_{C}$ are  $\mathcal{CN}\left(0,\sigma^{2}\right)$. Depending on the values of $H_A$, $H_B$ and $H_C$, R chooses a many to one map $\mathcal{M}^{H_A, H_B, H_C} : \mathcal{S}^3 \rightarrow \mathcal{S}^{'} $ where $\mathcal{S}^{'}$ is a signal set of size between $4^{2}$ and $4^{3}$ used by R during \textit{BC} phase. Notice that the minimum required size for $\mathcal{S}^{'}$ is 16, since 4 bits about the other two users needs to be conveyed to each of A, B and C.  \\ 

The elements in $\mathcal{S}^3 $ which are mapped to the same signal point in $\mathcal{S}^{'}$ by the map $\mathcal{M}^{H_A, H_B, H_C}$ are said to form a cluster. Let $\left\{\mathcal{L}_{1}, \mathcal{L}_{2},.., \mathcal{L}_{l}\right\}$ denote the set of all such clusters. The formation of clusters is called clustering, denoted by $\mathcal{C}^{H_A, H_B, H_C}$. \\

\begin{figure*}
\center
\includegraphics[height=40mm]{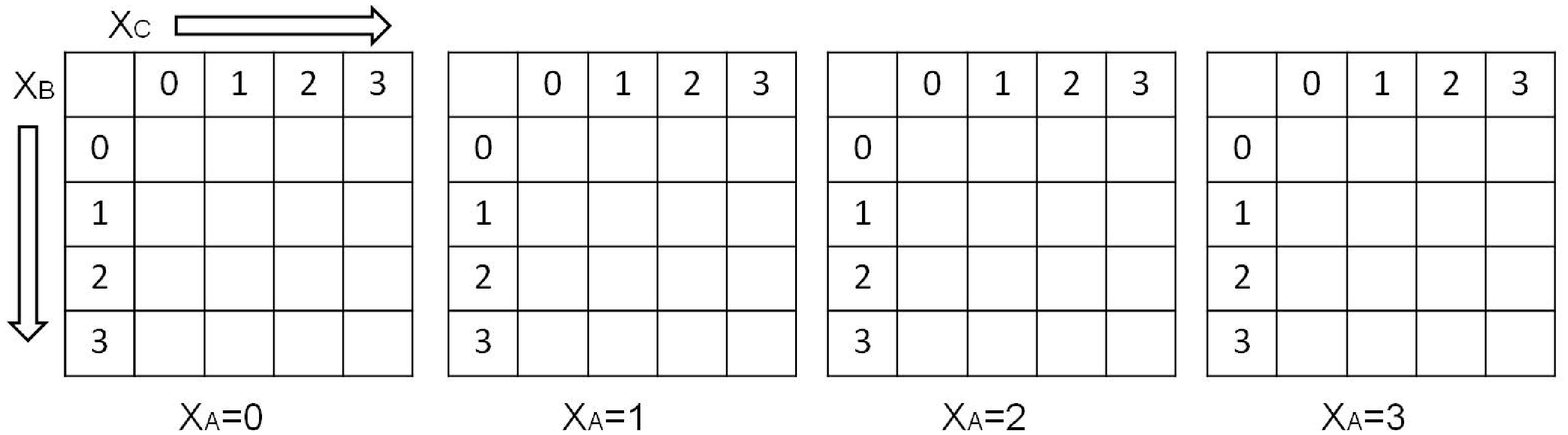}
\caption{The mapping observed at the Relay can be viewed as a Latin Cube of Second Order}
\end{figure*}

\begin{definition}
The cluster distance between a pair of clusters $\mathcal{L}_i$ and $\mathcal{L}_j$ is the minimum among all the distances calculated between the points $\left(x_{A}, x_{B}, x_{C}\right) \in \mathcal{L}_{i}$ and $\left(\acute{x_{A}}, \acute{x_{B}}, \acute{x_{C}}\right) \in \mathcal{L}_{j}$ in the effective constellation seen at the relay node R, as given in (\ref{cl1}) above.\\

\end{definition}

\begin{definition}
The \textit{minimum cluster distance} of the clustering $\mathcal{C}^{H_A, H_B, H_C}$ is the minimum among all the cluster distances, as given in (\ref{cl2}) at the top of this page.\\
\end{definition}

In order to ensure that A(/B/C) is able to decode B's and C's(/A's and C's /A's and B's) message, the clustering $\mathcal{C}$ should satisfy the exclusive law, as given in (\ref{mel1}), (\ref{mel2}), (\ref{mel3}) above.\\

The minimum cluster distance determines the performance during the MA phase of relaying. The performance during the BC phase is determined by the minimum distance of the signal set $\mathcal{S}^{'}$. For values of $\left(H_A, H_B, H_C\right)$ in the neighborhood of the singular fade states, the value of $d_{min}\left(\mathcal{C}^{H_A, H_B, H_C}\right)$ is greatly reduced, a phenomenon referred to as \textit{distance shortening} [4]. To avoid distance shortening, for each singular fade coefficient, a clustering needs to be chosen such that the minimum cluster distance at the singular fade state is non zero and is also maximized. \\

A clustering $\mathcal{C}^{H_A, H_B, H_C}$ is said to remove singular fade state $\left(H_A, H_B, H_C\right) \in \mathcal{H}, $ if $d_{min}\left(\mathcal{C}^{H_A, H_B, H_C}\right)>0$. For a singular fade state $\left(H_A, H_B, H_C\right) \in \mathcal{H} $, let $\mathcal{C}^{\left\{\left(H_A, H_B, H_C\right)\right\}} $ denote the clustering which removes the singular fade state $\left(H_A, H_B, H_C\right)$ (if there are multiple clusterings which remove the same singular fade state $\left(H_A, H_B, H_C\right)$, consider a clustering which maximizes the minimum cluster distance). Let $\mathcal{C_{H}}=\left\{\mathcal{C}^{\left\{\left(H_A, H_B, H_C\right)\right\}} : \left(H_A, H_B, H_C\right)\in \mathcal{H}\right\} $ denote the set of all such clusterings. \\

\begin{definition}
The minimum cluster distance of the clustering $\mathcal{C}^{h_A, h_B, h_C}$, when the fade state $(H_A, H_B, H_C)$ occurs in the MA phase, denoted by $d_{min}\left(\mathcal{C}^{h_A, h_B, h_C},H_A, H_B, H_C\right)$, is the minimum among all its cluster distances, as given in (\ref{cl3}).\\
\end{definition}

For $\left(H_A, H_B, H_C\right) \notin \mathcal{H}, $ the clustering $\mathcal{C}^{H_A, H_B, H_C} $ is chosen to be $\mathcal{C}^{\left\{\left(h_A, h_B, h_C\right)\right\}}$, which satisfies $ d_{min}\left(\mathcal{C}^{\left\{\left(h_A, h_B, h_C\right)\right\}},H_A, H_B, H_C\right) \geq d_{min}\left(\mathcal{C}^{\left\{\left(h'_A, h'_B, h'_C\right)\right\}},H_A, H_B, H_C\right), \forall \left(h_{A}, h_{B}, h_C\right) \neq \left(h'_{A}, h'_{B}, h'_{C}\right) \in \mathcal{H}.$
In [7], such clusterings that remove singular fade states are obtained with the help of Latin Squares while concentrating only on the first minimum cluster distance.  The clustering used by the relay is indicated to A, B and C using overhead bits.\\

The contributions of this paper are as follows:
\begin{itemize}
\item Using our proposed scheme, exchange of information in the wireless three-way relaying scenario is made possible with totally two channel uses.
\item It is shown that if the three users A, B, C transmit points from the same 4-PSK constellation, the requirement of satisfying the exclusive law is same as the clustering being represented by a Latin Cube of second order of side 4. To the best of our knowledge, this is the first work with only two channel uses for the three-way relaying scenario.
\item The singular fade states for the three user case are identified.
\item Clusterings that removes these singular fade states are obtained, that result in the size of the constellation used by the relay node R in the BC phase to lie between 16 to 23.
\item Simulation results are provided to verify that the adaptive clustering as obtained in the paper indeed performs better than non-adaptive clustering.
\end{itemize}
The remaining content is organized as follows: Section II demonstrates how a Latin Cube of Second Order and side 4 can be utilized to represent the network code for three user communication. In Section III the singular fade states are specified and in Section IV, clusterings corresponding to removal of each singular fade state are obtained using Latin Cubes of Second Order. Simulation results are shown in Section V. Section VI concludes the paper. 
\section{The Exclusive Law and Latin Cubes}

The nodes A, B and C transmit symbols from the same constellation, viz., 4-PSK. Our aim is to find the map that relay node R should use in order to cluster the $4^{3}$ possibilities of $\left(x_{A}, x_{B}, x_{C}\right)$ such that the exclusive law (given by (\ref{mel1}), (\ref{mel2}), (\ref{mel3})) is satisfied. The size of the constellation, or the number of clusters of the clustering that the relay utilizes has to be at least 16, since each user needs the 4 bit information corresponding to the other two users. Consider a $4 \times 4 \times 4$ array, whose 64 entries are indexed by $\left(x_{A}, x_{B}, x_{C}\right)$, i.e. the three messages that A, B and C send in the MA phase. Each file of this $4 \times 4 \times 4$ array, is indexed by a single value of $x_{A}$. Each row (column) of each file is indexed by a value of $x_B$ ($x_C$), for a fixed value of $x_A$. Now, a repetition of a symbol in a file results in the failure of exclusive law given by (\ref{mel1}). Consider the $4 \times 4$ array with its rows being the first(/second/third/forth) rows of the $4 \times 4 \times 4 $ array. Each $4 \times 4 $ array so obtained, corresponds to a single value of $x_{B}$. A repetition of a symbol in this array will result in the failure of exclusive law given by (\ref{mel2}). Similarly, consider the $4 \times 4$ array with its columns being the first(/second/third/forth) columns of the $4 \times 4 \times 4 $ array. Each $4 \times 4$ array so obtained, corresponds to a single value of $x_{C}$. A repetition of a symbol in this array will result in the failure of exclusive law given by (\ref{mel3}). Hence, if the exclusive law needs to be satisfied, then the cells of this array should be filled such that the $4 \times 4 \times 4 $ array so obtained, is a Latin cube of second order, for $t \geq 16$ (Definition 1). The clusters are obtained by putting together all the tuples $\left(i,j,k\right), i,j,k \in 0,1,...t-1 $ such that the entry in the $\left(i,j,k\right)$-th slot is the same entry from $\mathbb{Z}_{t}$. From above, we can say that all the relay clusterings that satisfy the mutually exclusive law forms Latin Cubes of second order of side 4 for $t \geq 16$, when the end nodes use 4-PSK constellations. Hence, now onwards, we consider the network code used by the relay node in the BC phase to be a $4 \times 4 \times 4$ array with files(/rows/columns) being indexed by the constellation point used by A(/B/C), symbols from the set $ \mathbb{Z}_{4}$ (Fig. 3). The cells of the array will be filled with elements of $\mathbb{Z}_{t}$ in such a way, that the resulting array is a Latin Cube of Second Order of side 4 and $t \geq 16$. Any arbitrary but unique symbol from $\mathbb{Z}_{t}$ denotes a unique cluster of a particular clustering.\\

%Consider a $16 \times 16$ array, that consists of the 16 possibilities of $\left(x_{A}, x_{B}\right)$, i.e. the pair of messages that A and B send at the first channel use, along the rows, and the 16 possibilities of $\left(x_{A}, x_{C}\right)$, i.e. the pair of messages that A and C send at the second channel use, along the columns, as shown in (Fig. 3(a)). Since A is assumed to have sent the same symbol in both the channel uses, it is required to consider only a part of this array, represented by the four diagonal $4 \times 4$ subarrays (Fig. 3(b)). 
\section{singular fade subspaces}

We earlier stated in Section II, that a clustering $\mathcal{C}^{H_A, H_B, H_C}$ is said to remove singular fade state $\left(H_A, H_B, H_C\right) \in \mathcal{H}, $ if $d_{min}\left(\mathcal{C}^{H_A, H_B, H_C}\right)>0$. Alternatively, removing singular fade states for a three-way relay channel can also be defined as follows:\\

\begin{definition}
A clustering $\mathcal{C}^{H_A, H_B, H_C}$ is said to \textit{remove the singular fade state} $\left(H_A, H_B, H_C\right) \in \mathcal{H}$, if any two possibilities of the messages sent by the users $\left(x_{A},x_{B},x_{C}\right), \left(x'_{A},x'_{B},x'_{C}\right) \in \mathcal{S}^{3}$ that satisfy\\
$$H_A x_A+ H_Bx_B+ H_C x_C=H_A x'_A+ H_B x'_B+ H_C x'_C
$$
are placed together in the same cluster by the clustering.\\
\end{definition}

\begin{definition}
A set $\left\{(x_A, x_B, x_C)\right\} \in \mathcal{S}^3 $ consisting of all the possibilities of $(x_A, x_B, x_C)$ that must be placed in the same cluster of the clustering used at relay node R in the BC phase in order to remove the singular fade state $\left(H_A, H_B, H_C\right)$ is referred to as a \textit{Singularity Removal Constraint} for the fade state $\left(H_A, H_B, H_C\right)$ for three-way relaying scenario.\\
\end{definition}

Let $(H_A, H_B, H_C) $ be the fade coefficient in the MA phase. The work in [4] and [7] shows that for the two-way relaying scenario, the $4^{2}$ possible pairs of symbols from 4-PSK constellation sent by the two users in the MA phase, can be clustered into a clustering dependent on a singular fade coefficient, of size 4 or 5 in a manner so as to remove this singular fade coefficient. In the case of three users, at the end of MA phase, relay receives a complex number, given by (\ref{yr}). Instead of R transmitting a point from the $4^{3}$ point constellation resulting from all the possibilities of $\left(x_{A}, x_{B}, x_{C}\right)$, the relay R can choose to group these  possibilities into clusters represented by a smaller constellation. We describe one such clustering in the following.\\

Let ${\Gamma}$ denote a singularity removal constraint corresponding to the singular fade state $(H_A, H_B, H_C) $ and let $(x_A, x_B, x_C), (x'_A, x'_B, x'_C) \in \mathcal{C} $. Then,
\begin{align}
\nonumber
& H_A x_A+ H_Bx_B+ H_C x_C=H_A x'_A+ H_B x'_B+ H_C x'_C \\
\nonumber
\Rightarrow & H_A (x_A-x'_A)+ H_B (x_B-x'_B)+ H_C (x_C-x'_C)=0 \\
\label{nullspace}
\Rightarrow & (H_A, H_B, H_C) \in 
\left\langle \left[ {\begin{array}{cc}
\vspace{0.15cm}
x_{A}-x'_A \\
\vspace{0.15cm}
x_{B}-x'_B \\
x_{C}-x'_C \\
\end{array} } \right]\right\rangle ^{\bot}
\end{align}where for a $3 \times 1$ non-zero vector $v$ over $\mathbb{C}$, $$\left\langle v\right\rangle ^{\bot}=\left\{w=(w_1,w_2,w_3) ~|~ w_1v_1+w_2v_2+w_{3} v_{3}=0 \right\}.$$ Clearly, $\left\langle v\right\rangle ^{\bot}$ is a two-dimensional vector space of $\mathbb{C}^3$. These values of $x_A, x_B, x_C, x'_A, x'_B, x'_C \in \mathcal{S}$ result in only finitely many possibilities for the right-hand side, since $\mathcal{S}$ is finite. Thus the singular fade states $(H_A, H_B, H_C),$ which are uncountably infinite, are points in a finite number of vector subspaces of $\mathbb{C}^3$. Henceforth, we shall refer to these finite number of vector subspaces as the \textit{Singular Fade Subspaces}. More precisely, there are three possibilities of singular fade subspaces for the three-way relaying as we explain individually in the following three cases.\\\\
\textit{Case 1:} One of the following subcases arise:
\begin{enumerate}
	\item $x_A=x'_A,~ x_B=x'_B \text{~and~} x_C \neq x'_C$
	\item $x_A=x'_A,~ x_B \neq x'_B \text{~and~} x_C=x'_C$
	\item $x_A \neq x'_A,~ x_B=x'_B \text{~and~} x_C=x'_C$
\end{enumerate}
\textit{Case 2:} One of the following subcases arise:
\begin{enumerate}
	\item $x_A=x'_A,~ x_B \neq x'_B \text{~and~} x_C \neq x'_C$
	\item $x_A \neq x'_A,~ x_B=x'_B \text{~and~} x_C \neq x'_C$
	\item $x_A \neq x'_A,~ x_B \neq x'_B \text{~and~} x_C=x'_C$
\end{enumerate}
\textit{Case 3:} $x_A \neq x'_A,~ x_B \neq x'_B \text{~and~} x_C \neq x'_C$\\\\
\textbf{\textit{Case 1:}} Without loss of generality, we discuss the third subcase of Case 1, i.e., the case when $x_A \neq x'_A,~ x_B=x'_B \text{~and~} x_C=x'_C$. The singular fade subspace in this case is given by 
$\mathcal{S}'=\left\langle \left[ {\begin{array}{cc}
\vspace{0.15cm}
x_{A}-x'_A \\
\vspace{0.15cm}
0 \\
0 \\
\end{array} } \right]\right\rangle ^{\bot}. $ 

The set of differences of the points of $\mathcal{S}$ is given by, $$\mathcal{D}=\left\{x_i-x_j ~|~ x_i, x_j \in \mathcal{S} \right\} = \left\{0, \pm 1\pm j, \pm2j, \pm 2\right\}. $$ Let $\mathcal{D}_{1}=\left\{ \pm1\pm j\right\}$ and $\mathcal{D}_{2}=\left\{ \pm2j, \pm2\right\}$. Then, $$\mathcal{D}=\left\{0\right\} \cup \mathcal{D}_{1} \cup \mathcal{D}_{2}.$$
Thus, $x_A- x'_A \in \mathcal{D}$ can take eight non-zero values. As a result, there are eight total possibilities for the vector $ \left[x_A-x'_A, ~0, ~0\right]^t $, where $v^t$ denotes the transpose of a vector $v$. Also, each one of $\pm 1\pm j$, $ \pm2j $ and $ \pm2$ can be obtained as scalar multiples of $1+j$ (over $\mathbb{C}$). Thus,\\\\
$ \footnotesize{
\left\langle \left[ {\begin{array}{cc}
\vspace{0.15cm}
1+j \\
\vspace{0.15cm}
0 \\
0 \\
\end{array} } \right]\right\rangle
=\left\langle \left[ {\begin{array}{cc}
\vspace{0.15cm}
\pm 1\pm j \\
\vspace{0.15cm}
0 \\
0 \\
\end{array} } \right]\right\rangle
=\left\langle \left[ {\begin{array}{cc}
\vspace{0.15cm}
\pm 2j \\
\vspace{0.15cm}
0 \\
0 \\
\end{array} } \right]\right\rangle 
=\left\langle \left[ {\begin{array}{cc}
\vspace{0.15cm}
\pm 2 \\
\vspace{0.15cm}
0 \\
0 \\
\end{array} } \right]\right\rangle.}\\$

So, $\exists$ only one singular fade subspace for the subcase, viz., $$\mathcal{S}'=\left\langle \left[ {\begin{array}{cc}
\vspace{0.15cm}
1+j \\
\vspace{0.15cm}
0 \\
0 \\
\end{array} } \right]\right\rangle ^{\bot}. $$ 
Similarly, for the other two subcases, there is one singular fade subspace, resulting in a total of 3 singular fade subspaces for the case.\\\\
\textbf{\textit{Case 2:}} Without loss of generality, we discuss the third subcase of Case 2, i.e., the case when $x_A \neq x'_A,~ x_B \neq x'_B \text{~and~} x_C=x'_C$. The singular fade subspace in this case is given by 
$\mathcal{S}''=\left\langle \left[ {\begin{array}{cc}
\vspace{0.15cm}
x_{A}-x'_A \\
\vspace{0.15cm}
x_B-x'_B \\
0 \\
\end{array} } \right]\right\rangle ^{\bot}.$

Here $x_A- x'_A \text{~and~} x_B-x'_B \in \mathcal{D}$ can take eight non-zero values each. There are therefore, 64 total possibilities for the vector $ \left[x_A-x'_A, ~x_B-x'_B, ~0\right]^t $. \\

\begin{lemma} For the case when $x_A- x'_A\neq 0,~  x_B-x'_B \neq 0 \text{~and~} x_C-x'_C = 0$, for a given vector $v= \left[x_A-x'_A, ~x_B-x'_B, ~0\right]^t$ over $\mathcal{D}_1 \cup \mathcal{D}_2$, there are precisely 4 or 8 vectors (including $v$) over $\mathcal{D}_1 \cup \mathcal{D}_2$ that generate the same vector space over $\mathbb{C}$ as $v$.
\end{lemma}
\begin{proof} As given in Section II of \cite{NMR}, for the 4-PSK constellation $\mathcal{S}$, the difference constellation $\mathcal{D}= \Delta \mathcal{S}=\left\{s-s': s,s' \in \mathcal{S}\right\}$ is of the form,
\begin{align}
\nonumber
\Delta \mathcal{S}=\left\{0\right\} & \cup \left\{2sin\left( \pi n / 4\right) e^{j k \pi / 2} | n \text{~odd}\right\}\\
\nonumber
&\cup \left\{2sin\left( \pi n / 4\right) e^{j\left( k \pi / 2 + \pi / 4\right)} | n \text{~even}\right\},
\end{align}
where $ 1 \leq n \leq 2 $ and $ 0 \leq k \leq 3 $. Therefore, we can write,
$$v= \left[ {\begin{array}{cc}
\vspace{0.15cm}
x_{A}-x'_A \\
\vspace{0.15cm}
x_B-x'_B \\
0 \\
\end{array} } \right] = \left[ {\begin{array}{cc}
\vspace{0.15cm}
2 sin \frac{ \pi k_{1}}{4} e^{j \phi_{1}} \\
\vspace{0.15cm}
2 sin \frac{ \pi k_{2}}{4} e^{j \phi_{2}} \\
0 \\
\end{array} } \right] $$
where $\phi_{i}=k_{i} \pi /2$ if $k_{i}$ is odd and $\phi_{i}=k_{i} \pi /2 + \pi/4$ if $k_{i}$ is even. \\

A vector $w$ over $\mathcal{D}_1 \cup \mathcal{D}_2$ shall generate the same vector space over $\mathbb{C}$ iff $w$ is a scalar multiple of $v$, i.e. for some complex number $r e^{j \theta} \in \mathbb{C}$, 
$$v = r e^{j\theta} w \Rightarrow \left[ {\begin{array}{cc}
\vspace{0.15cm}
2 sin \frac{ \pi k_{1}}{4} e^{j \phi_{1}} \\
\vspace{0.15cm}
2 sin \frac{ \pi k_{2}}{4} e^{j \phi_{2}} \\
0 \\
\end{array} } \right] = r e^{j\theta} \left[ {\begin{array}{cc}
\vspace{0.15cm}
2 sin \frac{ \pi k_{3}}{4} e^{j \phi_{3}} \\
\vspace{0.15cm}
2 sin \frac{ \pi k_{4}}{4} e^{j \phi_{4}} \\
0 \\
\end{array} } \right] $$ 
where for $i=3,4$ $\phi_{i}=k_{i} \pi /2$ if $k_{i}$ is odd and $\phi_{i}=k_{i} \pi /2 + \pi/4$ if $k_{i}$ is even. \\

Then,
\begin{equation}
\label{firstcomp}
2 sin \frac{ \pi k_{1}}{4} e^{j \phi_{1}} = r e^{j\theta} \times 2 sin \frac{ \pi k_{3}}{4} e^{j \phi_{3}}
\end{equation}
and
\begin{equation}
\label{secondcomp}
2 sin \frac{ \pi k_{2}}{4} e^{j \phi_{2}} = r e^{j\theta} \times 2 sin \frac{ \pi k_{4}}{4} e^{j \phi_{4}}.
\end{equation}

Dividing (\ref{firstcomp}) by (\ref{secondcomp}) and taking modulus of both sides, we get 
\begin{equation}
\label{rel}
\frac{sin \frac{ \pi k_{1}}{4}}{sin \frac{ \pi k_{2}}{4}}  =  \frac{sin \frac{ \pi k_{3}}{4} }{sin \frac{ \pi k_{4}}{4} }
\end{equation}
As shown in \cite{NMR}, this is possible only if $k_1=k_3$ and $k_2=k_4$. Also, from (\ref{firstcomp}) and (\ref{secondcomp}) we have 
\begin{equation}
\label{div}
\frac{ sin \frac{ \pi k_{1}}{4}}{sin \frac{ \pi k_{2}}{4}} e^{j (\phi_{1} - \phi_{2})} = \frac{ sin \frac{ \pi k_{3}}{4}}{sin \frac{ \pi k_{4}}{4}} e^{j (\phi_{3} - \phi_{4})}
\end{equation}
From (\ref{rel}) and (\ref{div}), we have 
\begin{equation}
\label{div}
e^{j (\phi_{1} - \phi_{2})} =  e^{j (\phi_{3} - \phi_{4})}
\end{equation}
Note that here, the LHS is fixed. It therefore suffices to compute that for the fixed value of the LHS, the number of values that RHS takes. It can be verified, that for a fixed value of $\phi_{1} - \phi_{2}$, there are precisely four pair of values of $\phi_{3} $ and $\phi_{4}$ that result in the same value of $\phi_{3} - \phi_{4}$. We now look at the following possibilities:\\
\textit{Case 1:} $k_1=k_2$. 
Then, $k_3=k_4$, i.e., there are exactly two possibilities for $k_1$ and $k_2$, viz., $k_1=k_2=1$ and $k_1=k_2=2$. With two pairs of values for $k_3$ and $k_4$ and four pairs of values for $\phi_{3} $ and $\phi_{4}$, we have a total of eight set of values that $w$ can take. Hence, in this case, the vector space generated by $v$ can be generated by exactly eight other vectors over $\mathcal{D}_1 \cup \mathcal{D}_2$. 
\textit{Case 2:} $k_1 \neq k_2$. 
Then, $k_3 \neq k_4$, i.e., there is precisely one possibility for $k_1$ and $k_2$, viz., $k_1=k_3$ and $k_1=k_4$. With only one possible set of values for $k_3$ and $k_4$ and four pairs of values for $\phi_{3} $ and $\phi_{4}$, we have a total of four set of values that $w$ can take. Hence, in this case, the vector space generated by $v$ can be generated by exactly four other vectors over $\mathcal{D}_1 \cup \mathcal{D}_2$. 

\end{proof}

In this case we end up with 12 singular fade subspaces given by the null spaces of the space given on the next page in Figure \ref{fig:sfscase2}.\\
\begin{figure*}
\centering
\tiny
$1. \left\langle \left[ {\begin{array}{cc}
\vspace{0.15cm}
1+j \\
\vspace{0.15cm}
1+j \\
0 \\
\end{array} } \right]\right\rangle
=\left\langle \left[ {\begin{array}{cc}
\vspace{0.15cm}
-1-j \\
\vspace{0.15cm}
-1-j \\
0 \\
\end{array} } \right]\right\rangle
=\left\langle \left[ {\begin{array}{cc}
\vspace{0.15cm}
1-j \\
\vspace{0.15cm}
1-j \\
0 \\
\end{array} } \right]\right\rangle 
=\left\langle \left[ {\begin{array}{cc}
\vspace{0.15cm}
-1+j \\
\vspace{0.15cm}
-1+j \\
0 \\
\end{array} } \right]\right\rangle
=\left\langle \left[ {\begin{array}{cc}
\vspace{0.15cm}
2j \\
\vspace{0.15cm}
2j \\
0 \\
\end{array} } \right]\right\rangle
=\left\langle \left[ {\begin{array}{cc}
\vspace{0.15cm}
-2j \\
\vspace{0.15cm}
-2j \\
0 \\
\end{array} } \right]\right\rangle
=\left\langle \left[ {\begin{array}{cc}
\vspace{0.15cm}
2 \\
\vspace{0.15cm}
2 \\
0 \\
\end{array} } \right]\right\rangle 
=\left\langle \left[ {\begin{array}{cc}
\vspace{0.15cm}
-2 \\
\vspace{0.15cm}
-2 \\
0 \\

\end{array} } \right]\right\rangle $~~~~~~~~~~~~\\
$2. \left\langle \left[ {\begin{array}{cc}
\vspace{0.15cm}
1+j \\
\vspace{0.15cm}
-1-j \\
0 \\
\end{array} } \right]\right\rangle
=\left\langle \left[ {\begin{array}{cc}
\vspace{0.15cm}
-1-j \\
\vspace{0.15cm}
1+j \\
0 \\
\end{array} } \right]\right\rangle
=\left\langle \left[ {\begin{array}{cc}
\vspace{0.15cm}
-1+j \\
\vspace{0.15cm}
1-j \\
0 \\
\end{array} } \right]\right\rangle 
=\left\langle \left[ {\begin{array}{cc}
\vspace{0.15cm}
1-j \\
\vspace{0.15cm}
-1+j \\
0 \\
\end{array} } \right]\right\rangle
=\left\langle \left[ {\begin{array}{cc}
\vspace{0.15cm}
2j \\
\vspace{0.15cm}
-2j \\
0 \\
\end{array} } \right]\right\rangle
=\left\langle \left[ {\begin{array}{cc}
\vspace{0.15cm}
-2j \\
\vspace{0.15cm}
2j \\
0 \\
\end{array} } \right]\right\rangle
=\left\langle \left[ {\begin{array}{cc}
\vspace{0.15cm}
2 \\
\vspace{0.15cm}
-2 \\
0 \\
\end{array} } \right]\right\rangle 
=\left\langle \left[ {\begin{array}{cc}
\vspace{0.15cm}
-2 \\
\vspace{0.15cm}
2 \\
0 \\
\end{array} } \right]\right\rangle $\\
$3. \left\langle \left[ {\begin{array}{cc}
\vspace{0.15cm}
1+j \\
\vspace{0.15cm}
1-j \\
0 \\
\end{array} } \right]\right\rangle
=\left\langle \left[ {\begin{array}{cc}
\vspace{0.15cm}
-1-j \\
\vspace{0.15cm}
-1+j \\
0 \\
\end{array} } \right]\right\rangle
=\left\langle \left[ {\begin{array}{cc}
\vspace{0.15cm}
-1+j \\
\vspace{0.15cm}
1+j \\
0 \\
\end{array} } \right]\right\rangle 
=\left\langle \left[ {\begin{array}{cc}
\vspace{0.15cm}
1-j \\
\vspace{0.15cm}
-1-j \\
0 \\
\end{array} } \right]\right\rangle
=\left\langle \left[ {\begin{array}{cc}
\vspace{0.15cm}
2j \\
\vspace{0.15cm}
2 \\
0 \\
\end{array} } \right]\right\rangle
=\left\langle \left[ {\begin{array}{cc}
\vspace{0.15cm}
-2j \\
\vspace{0.15cm}
-2 \\
0 \\
\end{array} } \right]\right\rangle
=\left\langle \left[ {\begin{array}{cc}
\vspace{0.15cm}
-2 \\
\vspace{0.15cm}
2j \\
0 \\
\end{array} } \right]\right\rangle 
=\left\langle \left[ {\begin{array}{cc}
\vspace{0.15cm}
2 \\
\vspace{0.15cm}
-2j \\
0 \\
\end{array} } \right]\right\rangle  $ ~~~~\\
$ 4.\left\langle \left[ {\begin{array}{cc}
\vspace{0.15cm}
1+j \\
\vspace{0.15cm}
-1+j \\
0 \\
\end{array} } \right]\right\rangle
=\left\langle \left[ {\begin{array}{cc}
\vspace{0.15cm}
-1-j \\
\vspace{0.15cm}
1-j \\
0 \\
\end{array} } \right]\right\rangle
=\left\langle \left[ {\begin{array}{cc}
\vspace{0.15cm}
-1+j \\
\vspace{0.15cm}
-1-j \\
0 \\
\end{array} } \right]\right\rangle 
=\left\langle \left[ {\begin{array}{cc}
\vspace{0.15cm}
1-j \\
\vspace{0.15cm}
1+j \\
0 \\
\end{array} } \right]\right\rangle 
=\left\langle \left[ {\begin{array}{cc}
\vspace{0.15cm}
2j \\
\vspace{0.15cm}
-2 \\
0 \\
\end{array} } \right]\right\rangle
=\left\langle \left[ {\begin{array}{cc}
\vspace{0.15cm}
-2j \\
\vspace{0.15cm}
2 \\
0 \\
\end{array} } \right]\right\rangle
=\left\langle \left[ {\begin{array}{cc}
\vspace{0.15cm}
2 \\
\vspace{0.15cm}
2j \\
0 \\
\end{array} } \right]\right\rangle 
=\left\langle \left[ {\begin{array}{cc}
\vspace{0.15cm}
-2 \\
\vspace{0.15cm}
-2j \\
0 \\
\end{array} } \right]\right\rangle $ ~~~~\\
$ 5. \left\langle \left[ {\begin{array}{cc}
\vspace{0.15cm}
1+j \\
\vspace{0.15cm}
2j \\
0 \\
\end{array} } \right]\right\rangle
=\left\langle \left[ {\begin{array}{cc}
\vspace{0.15cm}
-1-j \\
\vspace{0.15cm}
-2j \\
0 \\
\end{array} } \right]\right\rangle
=\left\langle \left[ {\begin{array}{cc}
\vspace{0.15cm}
-1+j \\
\vspace{0.15cm}
2 \\
0 \\
\end{array} } \right]\right\rangle 
=\left\langle \left[ {\begin{array}{cc}
\vspace{0.15cm}
1-j \\
\vspace{0.15cm}
-2 \\
0 \\
\end{array} } \right]\right\rangle $ ~~~~
$6. \left\langle \left[ {\begin{array}{cc}
\vspace{0.15cm}
1+j \\
\vspace{0.15cm}
-2j \\
0 \\
\end{array} } \right]\right\rangle
=\left\langle \left[ {\begin{array}{cc}
\vspace{0.15cm}
-1-j \\
\vspace{0.15cm}
2j \\
0 \\
\end{array} } \right]\right\rangle
=\left\langle \left[ {\begin{array}{cc}
\vspace{0.15cm}
-1+j \\
\vspace{0.15cm}
-2 \\
0 \\
\end{array} } \right]\right\rangle 
=\left\langle \left[ {\begin{array}{cc}
\vspace{0.15cm}
1-j \\
\vspace{0.15cm}
2 \\
0 \\
\end{array} } \right]\right\rangle $ \\
$ 7. \left\langle \left[ {\begin{array}{cc}
\vspace{0.15cm}
1+j \\
\vspace{0.15cm}
2 \\
0 \\
\end{array} } \right]\right\rangle
=\left\langle \left[ {\begin{array}{cc}
\vspace{0.15cm}
-1-j \\
\vspace{0.15cm}
-2 \\
0 \\
\end{array} } \right]\right\rangle
=\left\langle \left[ {\begin{array}{cc}
\vspace{0.15cm}
-1+j \\
\vspace{0.15cm}
2j \\
0 \\
\end{array} } \right]\right\rangle 
=\left\langle \left[ {\begin{array}{cc}
\vspace{0.15cm}
1-j \\
\vspace{0.15cm}
-2j \\
0 \\
\end{array} } \right]\right\rangle $ ~~~~
$8. \left\langle \left[ {\begin{array}{cc}
\vspace{0.15cm}
1+j \\
\vspace{0.15cm}
-2 \\
0 \\
\end{array} } \right]\right\rangle
=\left\langle \left[ {\begin{array}{cc}
\vspace{0.15cm}
-1-j \\
\vspace{0.15cm}
2 \\
0 \\
\end{array} } \right]\right\rangle
=\left\langle \left[ {\begin{array}{cc}
\vspace{0.15cm}
-1+j \\
\vspace{0.15cm}
-2j \\
0 \\
\end{array} } \right]\right\rangle 
=\left\langle \left[ {\begin{array}{cc}
\vspace{0.15cm}
1-j \\
\vspace{0.15cm}
2j \\
0 \\
\end{array} } \right]\right\rangle $ \\
$9. \left\langle \left[ {\begin{array}{cc}
\vspace{0.15cm}
2j \\
\vspace{0.15cm}
1+j \\
0 \\
\end{array} } \right]\right\rangle
=\left\langle \left[ {\begin{array}{cc}
\vspace{0.15cm}
-2j \\
\vspace{0.15cm}
-1-j \\
0 \\
\end{array} } \right]\right\rangle
=\left\langle \left[ {\begin{array}{cc}
\vspace{0.15cm}
2 \\
\vspace{0.15cm}
-1+j \\
0 \\
\end{array} } \right]\right\rangle 
=\left\langle \left[ {\begin{array}{cc}
\vspace{0.15cm}
-2 \\
\vspace{0.15cm}
1-j \\
0 \\
\end{array} } \right]\right\rangle $ ~~~~
$ 10. \left\langle \left[ {\begin{array}{cc}
\vspace{0.15cm}
-2j \\
\vspace{0.15cm}
1+j \\
0 \\
\end{array} } \right]\right\rangle
=\left\langle \left[ {\begin{array}{cc}
\vspace{0.15cm}
2j \\
\vspace{0.15cm}
-1-j \\
0 \\
\end{array} } \right]\right\rangle
=\left\langle \left[ {\begin{array}{cc}
\vspace{0.15cm}
-2 \\
\vspace{0.15cm}
-1+j \\
0 \\
\end{array} } \right]\right\rangle 
=\left\langle \left[ {\begin{array}{cc}
\vspace{0.15cm}
2 \\
\vspace{0.15cm}
1-j \\
0 \\
\end{array} } \right]\right\rangle $ \\
$11. \left\langle \left[ {\begin{array}{cc}
\vspace{0.15cm}
2 \\
\vspace{0.15cm}
1+j \\
0 \\
\end{array} } \right]\right\rangle
=\left\langle \left[ {\begin{array}{cc}
\vspace{0.15cm}
-2 \\
\vspace{0.15cm}
-1-j \\
0 \\
\end{array} } \right]\right\rangle
=\left\langle \left[ {\begin{array}{cc}
\vspace{0.15cm}
2j \\
\vspace{0.15cm}
-1+j \\
0 \\
\end{array} } \right]\right\rangle 
=\left\langle \left[ {\begin{array}{cc}
\vspace{0.15cm}
-2j \\
\vspace{0.15cm}
1-j \\
0 \\
\end{array} } \right]\right\rangle $ ~~~
$12. \left\langle \left[ {\begin{array}{cc}
\vspace{0.15cm}
-2 \\
\vspace{0.15cm}
1+j \\
0 \\
\end{array} } \right]\right\rangle
=\left\langle \left[ {\begin{array}{cc}
\vspace{0.15cm}
2 \\
\vspace{0.15cm}
-1-j \\
0 \\
\end{array} } \right]\right\rangle
=\left\langle \left[ {\begin{array}{cc}
\vspace{0.15cm}
-2j \\
\vspace{0.15cm}
-1+j \\
0 \\
\end{array} } \right]\right\rangle 
=\left\langle \left[ {\begin{array}{cc}
\vspace{0.15cm}
2j \\
\vspace{0.15cm}
1-j \\
0 \\
\end{array} } \right]\right\rangle. $
\label{fig:sfscase2}
\caption{Null Spaces of the Singular Fades Spaces for the case $x_A \neq x'_A, ~ x_B \neq x'_B \text{~and~} x_C=x'_C$.}
\end{figure*}

So, $\exists$ 12 singular fade subspaces for the subcase each being the null space of the above 12 spaces. Similarly, for each of the other two subcases, there are 12 singular fade subspaces, resulting in a total of 36 singular fade subspaces for the case.\\\\
\textbf{\textit{Case 3:}} In this case, $x_A \neq x'_A,~ x_B \neq x'_B \text{~and~} x_C \neq x'_C$. The singular fade subspace in this case is given by 
\begin{equation}
\label{scase3}
\mathcal{S}'''=\left\langle \left[ {\begin{array}{cc}
\vspace{0.15cm}
x_{A}-x'_A \\
\vspace{0.15cm}
x_B-x'_B \\
x_C-x'_C \\
\end{array} } \right]\right\rangle ^{\bot}.
\end{equation}

Here each of $x_A- x'_A,~ x_B-x'_B \text{~and~} x_C-x'_C \in \mathcal{D}$ can take eight non-zero values. There are therefore, 512 total possibilities for the vector $ \left[x_A-x'_A, ~x_B-x'_B, ~x_C-x'_C\right]^t $. In this case we end up with 112 singular fade subspaces. We now explain this.\\

\begin{lemma} For the case when $x_A- x'_A\neq 0,~  x_B-x'_B \neq 0 \text{~and~} x_C-x'_C \neq 0$, for a given vector $v= \left[x_A-x'_A, ~x_B-x'_B, ~x_C-x'_C\right]^t$ over $\mathcal{D}_1 \cup \mathcal{D}_2$, there are precisely 4 or 8 vectors (including $v$) over $\mathcal{D}_1 \cup \mathcal{D}_2$ that generate the same vector space over $\mathbb{C}$ as $v$.
\end{lemma}
\begin{proof}
%The proof for this lemma is similar to \textit{Lemma 1} and is omitted. Here, the vector space generated by $v$ can be generated by exactly four other vectors over $\mathcal{D}_1 \cup \mathcal{D}_2$ if one or two of $x_A- x'_A,~  x_B-x'_B \text{~and~} x_C-x'_C $ belong to $\mathcal{D}_2$; and the vector space generated by $v$ can be generated by exactly eight other vectors over $\mathcal{D}_1 \cup \mathcal{D}_2$ if all of $x_A- x'_A,~  x_B-x'_B \text{~and~} x_C-x'_C $ belong to either $\mathcal{D}_1$ or $\mathcal{D}_2$.\\
As mentioned in the proof of Lemma 1, from \cite{NMR},
\begin{align}
\nonumber
\mathcal{D}=\Delta \mathcal{S}=\left\{0\right\} & \cup \left\{2sin\left( \pi n / 4\right) e^{j k \pi / 2} | n \text{~odd}\right\}\\
\nonumber
&\cup \left\{2sin\left( \pi n / 4\right) e^{j\left( k \pi / 2 + \pi / 4\right)} | n \text{~even}\right\},
\end{align}
where $ 1 \leq n \leq 2 $ and $ 0 \leq k \leq 3 $. Therefore, we can write,
$$v= \left[ {\begin{array}{cc}
\vspace{0.15cm}
x_{A}-x'_A \\
\vspace{0.15cm}
x_B-x'_B \\
x_C-x'_C \\
\end{array} } \right] = \left[ {\begin{array}{cc}
\vspace{0.15cm}
2 sin \frac{ \pi k_{1}}{4} e^{j \phi_{1}} \\
\vspace{0.15cm}
2 sin \frac{ \pi k_{2}}{4} e^{j \phi_{2}} \\
2 sin \frac{ \pi k_{3}}{4} e^{j \phi_{3}} \\
\end{array} } \right] $$
where $\phi_{i}=k_{i} \pi /2$ if $k_{i}$ is odd and $\phi_{i}=k_{i} \pi /2 + \pi/4$ if $k_{i}$ is even. \\

A vector $w$ over $\mathcal{D}_1 \cup \mathcal{D}_2$ shall generate the same vector space over $\mathbb{C}$ iff $w$ is a scalar multiple of $v$, i.e. for some complex number $r e^{j \theta} \in \mathbb{C}$, 
$$v = r e^{j\theta} w \Rightarrow \left[ {\begin{array}{cc}
\vspace{0.15cm}
2 sin \frac{ \pi k_{1}}{4} e^{j \phi_{1}} \\
\vspace{0.15cm}
2 sin \frac{ \pi k_{2}}{4} e^{j \phi_{2}} \\
2 sin \frac{ \pi k_{3}}{4} e^{j \phi_{3}} \\
\end{array} } \right] = r e^{j\theta} \left[ {\begin{array}{cc}
\vspace{0.15cm}
2 sin \frac{ \pi k_{4}}{4} e^{j \phi_{4}} \\
\vspace{0.15cm}
2 sin \frac{ \pi k_{5}}{4} e^{j \phi_{5}} \\
2 sin \frac{ \pi k_{6}}{4} e^{j \phi_{6}} \\
\end{array} } \right] $$ 
where for $i=4,5,6$ $\phi_{i}=k_{i} \pi /2$ if $k_{i}$ is odd and $\phi_{i}=k_{i} \pi /2 + \pi/4$ if $k_{i}$ is even. \\

Then,
\begin{equation}
\label{firstcomp2}
2 sin \frac{ \pi k_{1}}{4} e^{j \phi_{1}} = r e^{j\theta} \times 2 sin \frac{ \pi k_{4}}{4} e^{j \phi_{4}},
\end{equation}
\begin{equation}
\label{secondcomp2}
2 sin \frac{ \pi k_{2}}{4} e^{j \phi_{2}} = r e^{j\theta} \times 2 sin \frac{ \pi k_{5}}{4} e^{j \phi_{5}}
\end{equation}
and,
\begin{equation}
\label{thirdcomp2}
2 sin \frac{ \pi k_{3}}{4} e^{j \phi_{3}} = r e^{j\theta} \times 2 sin \frac{ \pi k_{6}}{4} e^{j \phi_{6}}
\end{equation}

Dividing (\ref{firstcomp2}) by (\ref{secondcomp2}) and taking modulus of both sides, we get 
\begin{equation}
\label{rel2}
\frac{sin \frac{ \pi k_{1}}{4}}{sin \frac{ \pi k_{2}}{4}}  =  \frac{sin \frac{ \pi k_{4}}{4} }{sin \frac{ \pi k_{5}}{4} }
\end{equation}
As shown in \cite{NMR}, this is possible only if $k_1=k_4$ and $k_2=k_5$. Similarly, dividing (\ref{firstcomp2}) by (\ref{thirdcomp2}) and taking modulus of both sides, we get 
\begin{equation}
\label{rel3}
\frac{sin \frac{ \pi k_{1}}{4}}{sin \frac{ \pi k_{3}}{4}}  =  \frac{sin \frac{ \pi k_{4}}{4} }{sin \frac{ \pi k_{6}}{4} }
\end{equation}
As shown in \cite{NMR}, this is possible only if $k_1=k_4$ and $k_3=k_6$. Also, from (\ref{firstcomp2}) and (\ref{secondcomp2}) we have 
\begin{equation}
\label{div2}
\frac{ sin \frac{ \pi k_{1}}{4}}{sin \frac{ \pi k_{2}}{4}} e^{j (\phi_{1} - \phi_{2})} = \frac{ sin \frac{ \pi k_{4}}{4}}{sin \frac{ \pi k_{5}}{4}} e^{j (\phi_{4} - \phi_{5})}
\end{equation}
From (\ref{rel2}) and (\ref{div2}), we have 
\begin{equation}
\label{div3}
e^{j (\phi_{1} - \phi_{2})} =  e^{j (\phi_{4} - \phi_{5})}
\end{equation}
Similarly, 
\begin{equation}
\label{div4}
e^{j (\phi_{1} - \phi_{3})} =  e^{j (\phi_{4} - \phi_{6})}
\end{equation}
In (\ref{div3}) and (\ref{div4}), the LHS is fixed. It therefore suffices to compute the number of values that RHS takes for fixed LHS in the two equations. It can be verified, that for a fixed value of $\phi_{1},~ \phi_{2} \text{~and~} \phi_3$, there are precisely four set of values of $\phi_{4},~ \phi_{5} \text{~and~} \phi_6$ that result in the same value of $\phi_{1} - \phi_{2}$ and $\phi_{1} - \phi_{3}$. We now look at the following possibilities:\\
\textit{Case 1:} $k_1=k_2=k_3$. 
Then, $k_4=k_5=k_6$, i.e., there are exactly two possibilities for $k_4$, $k_5$ and $k_6$, viz., $k_4=k_5=k_6=1$ and $k_4=k_5=k_6=2$. With two sets of values for $k_4$, $k_5$ and $k_6$ and four sets of values for $\phi_{4} $, $\phi_{5}$ and $\phi_6$, we have a total of eight set of values that $w$ can take. Hence, in this case, the vector space generated by $v$ can be generated by exactly eight other vectors over $\mathcal{D}_1 \cup \mathcal{D}_2$. 
\textit{Case 2:} Atleast one of $k_1 \neq k_2$, $k_1 \neq k_3$, or $k_2 \neq k_3$. 
Then, $k_4 \neq k_5$, $k_4 \neq k_6$, or $k_5 \neq k_6$, so that, using (\ref{rel2}) and (\ref{rel3}), there is precisely one possibility for $k_4,~k_5,~ k_6$, viz., $k_4=k_1$, $k_5=k_2$ and $k_6=k_3$. With only one possible set of values for $k_4,~k_5$ and $k_6$ and four sets of values for $\phi_{4}, ~ \phi_5 $ and $\phi_{6}$, we have a total of four set of values that $w$ can take. Hence, in this case, the vector space generated by $v$ can be generated by exactly four other vectors over $\mathcal{D}_1 \cup \mathcal{D}_2$. \\
\end{proof}

Since all of $x_A- x'_A,~ x_B-x'_B \text{~and~} x_C-x'_C \in \mathcal{D}$ are non-zero, we can say that $$x_A- x'_A,~ x_B-x'_B \text{~and~} x_C-x'_C \in \mathcal{D}_{1} \cup \mathcal{D}_2.$$ As a result, we have the following three subcases:
\begin{enumerate}
 \item One of $x_A- x'_A,~ x_B-x'_B \text{~and~} x_C-x'_C \in \mathcal{D}_1$
 \item Two of $x_A- x'_A,~ x_B-x'_B \text{~and~} x_C-x'_C \in \mathcal{D}_1$
 \item All of $x_A- x'_A,~ x_B-x'_B \text{~and~} x_C-x'_C \in \mathcal{D}_1$
\end{enumerate}
We deal with each one of the subcases one-by-one. \\\\
\textit{Subcase 1:} One of $x_A- x'_A,~ x_B-x'_B \text{~and~} x_C-x'_C \in \mathcal{D}_1$. Without loss of generality, we assume that $x_A-x'_A \in \mathcal{D}_1$ and $x_B-x'_B, ~ x_C-x'_C \in \mathcal{D}_2$. The singular fade subspace for the case is given by (\ref{scase3}). There are 64 possibilities for the vector $v'=\left[x_{A}-x'_A, ~ x_B-x'_B, ~ x_C-x'_C\right]^{t}$. But some of the possibilities may generate the same vector space over $\mathbb{C}$. There are precisely 4 vectors of length 3 over $\mathcal{D}_{1} \cup \mathcal{D}_2$, the $\left\{\pm1, \pm j \right\}$ scalar multiples of the vector. As a result, the case $x_A-x'_A \in \mathcal{D}_1$ and $x_B-x'_B, ~ x_C-x'_C \in \mathcal{D}_2$ leads to 16 singular fade subspaces as shown in Figure \ref{fig:sfssubcase1}. The same holds for the case when $x_B-x'_B \in \mathcal{D}_1$ and $x_A-x'_A, ~ x_C-x'_C \in \mathcal{D}_2$, or when $x_C-x'_C \in \mathcal{D}_1$ and $x_A-x'_A, ~ x_B-x'_B \in \mathcal{D}_2$. This subcase therefore results in 48 singular fade subspaces.\\

\begin{figure*}
\centering
\tiny
$ 1. \left\langle \left[ {\begin{array}{cc}
\vspace{0.15cm}
1+j \\
\vspace{0.15cm}
2j \\
2j \\
\end{array} } \right]\right\rangle
=\left\langle \left[ {\begin{array}{cc}
\vspace{0.15cm}
-1-j \\
\vspace{0.15cm}
-2j \\
-2j \\
\end{array} } \right]\right\rangle
=\left\langle \left[ {\begin{array}{cc}
\vspace{0.15cm}
-1+j \\
\vspace{0.15cm}
-2 \\
-2 \\
\end{array} } \right]\right\rangle 
=\left\langle \left[ {\begin{array}{cc}
\vspace{0.15cm}
1-j \\
\vspace{0.15cm}
2 \\
2 \\
\end{array} } \right]\right\rangle $ ~~~~
$ 2. \left\langle \left[ {\begin{array}{cc}
\vspace{0.15cm}
1+j \\
\vspace{0.15cm}
2j \\
-2j \\
\end{array} } \right]\right\rangle
=\left\langle \left[ {\begin{array}{cc}
\vspace{0.15cm}
-1-j \\
\vspace{0.15cm}
-2j \\
2j \\
\end{array} } \right]\right\rangle
=\left\langle \left[ {\begin{array}{cc}
\vspace{0.15cm}
-1+j \\
\vspace{0.15cm}
-2 \\
2 \\
\end{array} } \right]\right\rangle 
=\left\langle \left[ {\begin{array}{cc}
\vspace{0.15cm}
1-j \\
\vspace{0.15cm}
2 \\
-2 \\
\end{array} } \right]\right\rangle $ \\
$ 3. \left\langle \left[ {\begin{array}{cc}
\vspace{0.15cm}
1+j \\
\vspace{0.15cm}
2j \\
2 \\
\end{array} } \right]\right\rangle
=\left\langle \left[ {\begin{array}{cc}
\vspace{0.15cm}
-1-j \\
\vspace{0.15cm}
-2j \\
-2 \\
\end{array} } \right]\right\rangle
=\left\langle \left[ {\begin{array}{cc}
\vspace{0.15cm}
-1+j \\
\vspace{0.15cm}
-2 \\
2j \\
\end{array} } \right]\right\rangle 
=\left\langle \left[ {\begin{array}{cc}
\vspace{0.15cm}
1-j \\
\vspace{0.15cm}
2 \\
-2j \\
\end{array} } \right]\right\rangle $ ~~~~
$ 4. \left\langle \left[ {\begin{array}{cc}
\vspace{0.15cm}
1+j \\
\vspace{0.15cm}
2j \\
-2 \\
\end{array} } \right]\right\rangle
=\left\langle \left[ {\begin{array}{cc}
\vspace{0.15cm}
-1-j \\
\vspace{0.15cm}
-2j \\
2 \\
\end{array} } \right]\right\rangle
=\left\langle \left[ {\begin{array}{cc}
\vspace{0.15cm}
-1+j \\
\vspace{0.15cm}
-2 \\
-2j \\
\end{array} } \right]\right\rangle 
=\left\langle \left[ {\begin{array}{cc}
\vspace{0.15cm}
1-j \\
\vspace{0.15cm}
2 \\
2j \\
\end{array} } \right]\right\rangle $ \\
$ 5. \left\langle \left[ {\begin{array}{cc}
\vspace{0.15cm}
1+j \\
\vspace{0.15cm}
-2j \\
2j \\
\end{array} } \right]\right\rangle
=\left\langle \left[ {\begin{array}{cc}
\vspace{0.15cm}
-1-j \\
\vspace{0.15cm}
2j \\
-2j \\
\end{array} } \right]\right\rangle
=\left\langle \left[ {\begin{array}{cc}
\vspace{0.15cm}
-1+j \\
\vspace{0.15cm}
2 \\
-2 \\
\end{array} } \right]\right\rangle 
=\left\langle \left[ {\begin{array}{cc}
\vspace{0.15cm}
1-j \\
\vspace{0.15cm}
-2 \\
2 \\
\end{array} } \right]\right\rangle $ ~~~~
$ 6. \left\langle \left[ {\begin{array}{cc}
\vspace{0.15cm}
1+j \\
\vspace{0.15cm}
-2j \\
-2j \\
\end{array} } \right]\right\rangle
=\left\langle \left[ {\begin{array}{cc}
\vspace{0.15cm}
-1-j \\
\vspace{0.15cm}
2j \\
2j \\
\end{array} } \right]\right\rangle
=\left\langle \left[ {\begin{array}{cc}
\vspace{0.15cm}
-1+j \\
\vspace{0.15cm}
2 \\
2 \\
\end{array} } \right]\right\rangle 
=\left\langle \left[ {\begin{array}{cc}
\vspace{0.15cm}
1-j \\
\vspace{0.15cm}
-2 \\
-2 \\
\end{array} } \right]\right\rangle $ \\
$ 7. \left\langle \left[ {\begin{array}{cc}
\vspace{0.15cm}
2 1+j\\
\vspace{0.15cm}
-2j \\
2 \\
\end{array} } \right]\right\rangle
=\left\langle \left[ {\begin{array}{cc}
\vspace{0.15cm}
-1-j \\
\vspace{0.15cm}
2j \\
-2 \\
\end{array} } \right]\right\rangle
=\left\langle \left[ {\begin{array}{cc}
\vspace{0.15cm}
-1+j \\
\vspace{0.15cm}
2 \\
2j \\
\end{array} } \right]\right\rangle 
=\left\langle \left[ {\begin{array}{cc}
\vspace{0.15cm}
1-j \\
\vspace{0.15cm}
-2 \\
-2j \\
\end{array} } \right]\right\rangle $ ~~~
$ 8. \left\langle \left[ {\begin{array}{cc}
\vspace{0.15cm}
1+j \\
\vspace{0.15cm}
-2j \\
-2 \\
\end{array} } \right]\right\rangle
=\left\langle \left[ {\begin{array}{cc}
\vspace{0.15cm}
-1-j \\
\vspace{0.15cm}
2j \\
2 \\
\end{array} } \right]\right\rangle
=\left\langle \left[ {\begin{array}{cc}
\vspace{0.15cm}
-1+j \\
\vspace{0.15cm}
2 \\
-2j \\
\end{array} } \right]\right\rangle 
=\left\langle \left[ {\begin{array}{cc}
\vspace{0.15cm}
1-j \\
\vspace{0.15cm}
-2 \\
2j \\
\end{array} } \right]\right\rangle. $\\
$ 9. \left\langle \left[ {\begin{array}{cc}
\vspace{0.15cm}
1+j \\
\vspace{0.15cm}
2 \\
2j \\
\end{array} } \right]\right\rangle
=\left\langle \left[ {\begin{array}{cc}
\vspace{0.15cm}
-1-j \\
\vspace{0.15cm}
-2 \\
-2j \\
\end{array} } \right]\right\rangle
=\left\langle \left[ {\begin{array}{cc}
\vspace{0.15cm}
-1+j \\
\vspace{0.15cm}
2j \\
-2 \\
\end{array} } \right]\right\rangle 
=\left\langle \left[ {\begin{array}{cc}
\vspace{0.15cm}
1-j \\
\vspace{0.15cm}
-2j \\
2 \\
\end{array} } \right]\right\rangle $ ~~~~
$ 10. \left\langle \left[ {\begin{array}{cc}
\vspace{0.15cm}
1+j \\
\vspace{0.15cm}
2 \\
-2j \\
\end{array} } \right]\right\rangle
=\left\langle \left[ {\begin{array}{cc}
\vspace{0.15cm}
-1-j \\
\vspace{0.15cm}
-2 \\
2j \\
\end{array} } \right]\right\rangle
=\left\langle \left[ {\begin{array}{cc}
\vspace{0.15cm}
-1+j \\
\vspace{0.15cm}
2j \\
2 \\
\end{array} } \right]\right\rangle 
=\left\langle \left[ {\begin{array}{cc}
\vspace{0.15cm}
1-j \\
\vspace{0.15cm}
-2j \\
-2 \\
\end{array} } \right]\right\rangle $ \\
$ 11. \left\langle \left[ {\begin{array}{cc}
\vspace{0.15cm}
1+j \\
\vspace{0.15cm}
2 \\
2 \\
\end{array} } \right]\right\rangle
=\left\langle \left[ {\begin{array}{cc}
\vspace{0.15cm}
-1-j \\
\vspace{0.15cm}
-2 \\
-2 \\
\end{array} } \right]\right\rangle
=\left\langle \left[ {\begin{array}{cc}
\vspace{0.15cm}
-1+j \\
\vspace{0.15cm}
2j \\
2j \\
\end{array} } \right]\right\rangle 
=\left\langle \left[ {\begin{array}{cc}
\vspace{0.15cm}
1-j \\
\vspace{0.15cm}
-2j \\
-2j \\
\end{array} } \right]\right\rangle $ ~~~~
$ 12. \left\langle \left[ {\begin{array}{cc}
\vspace{0.15cm}
1+j \\
\vspace{0.15cm}
2 \\
-2 \\
\end{array} } \right]\right\rangle
=\left\langle \left[ {\begin{array}{cc}
\vspace{0.15cm}
-1-j \\
\vspace{0.15cm}
-2 \\
2 \\
\end{array} } \right]\right\rangle
=\left\langle \left[ {\begin{array}{cc}
\vspace{0.15cm}
-1+j \\
\vspace{0.15cm}
2j \\
-2j \\
\end{array} } \right]\right\rangle 
=\left\langle \left[ {\begin{array}{cc}
\vspace{0.15cm}
1-j \\
\vspace{0.15cm}
-2j \\
2j \\
\end{array} } \right]\right\rangle $ \\
$ 13. \left\langle \left[ {\begin{array}{cc}
\vspace{0.15cm}
1+j \\
\vspace{0.15cm}
-2 \\
2j \\
\end{array} } \right]\right\rangle
=\left\langle \left[ {\begin{array}{cc}
\vspace{0.15cm}
-1-j \\
\vspace{0.15cm}
2 \\
-2j \\
\end{array} } \right]\right\rangle
=\left\langle \left[ {\begin{array}{cc}
\vspace{0.15cm}
-1+j \\
\vspace{0.15cm}
-2j \\
-2 \\
\end{array} } \right]\right\rangle 
=\left\langle \left[ {\begin{array}{cc}
\vspace{0.15cm}
1-j \\
\vspace{0.15cm}
2j \\
2 \\
\end{array} } \right]\right\rangle $ ~~~~
$ 14. \left\langle \left[ {\begin{array}{cc}
\vspace{0.15cm}
1+j \\
\vspace{0.15cm}
-2 \\
2j \\
\end{array} } \right]\right\rangle
=\left\langle \left[ {\begin{array}{cc}
\vspace{0.15cm}
-1-j \\
\vspace{0.15cm}
2 \\
-2j \\
\end{array} } \right]\right\rangle
=\left\langle \left[ {\begin{array}{cc}
\vspace{0.15cm}
-1+j \\
\vspace{0.15cm}
-2j \\
-2 \\
\end{array} } \right]\right\rangle 
=\left\langle \left[ {\begin{array}{cc}
\vspace{0.15cm}
1-j \\
\vspace{0.15cm}
2j \\
2 \\
\end{array} } \right]\right\rangle $ \\
$ 15. \left\langle \left[ {\begin{array}{cc}
\vspace{0.15cm}
1+j \\
\vspace{0.15cm}
-2 \\
2 \\
\end{array} } \right]\right\rangle
=\left\langle \left[ {\begin{array}{cc}
\vspace{0.15cm}
-1-j \\
\vspace{0.15cm}
2 \\
-2 \\
\end{array} } \right]\right\rangle
=\left\langle \left[ {\begin{array}{cc}
\vspace{0.15cm}
-1+j \\
\vspace{0.15cm}
-2j \\
2j \\
\end{array} } \right]\right\rangle 
=\left\langle \left[ {\begin{array}{cc}
\vspace{0.15cm}
1-j \\
\vspace{0.15cm}
2j \\
-2j \\
\end{array} } \right]\right\rangle $ ~~~
$ 16. \left\langle \left[ {\begin{array}{cc}
\vspace{0.15cm}
1+j \\
\vspace{0.15cm}
-2 \\
-2 \\
\end{array} } \right]\right\rangle
=\left\langle \left[ {\begin{array}{cc}
\vspace{0.15cm}
-1-j \\
\vspace{0.15cm}
2 \\
2 \\
\end{array} } \right]\right\rangle
=\left\langle \left[ {\begin{array}{cc}
\vspace{0.15cm}
-1+j \\
\vspace{0.15cm}
-2j \\
-2j \\
\end{array} } \right]\right\rangle 
=\left\langle \left[ {\begin{array}{cc}
\vspace{0.15cm}
1-j \\
\vspace{0.15cm}
2j \\
2j \\
\end{array} } \right]\right\rangle. $
\label{fig:sfssubcase1}
\caption{Null Spaces of the Singular Fades Spaces for the case $x_A - x'_A \in \mathcal{D}_1$ and $x_B-x'_B, ~ x_C-x'_C \in \mathcal{D}_2$.}
\end{figure*}

\begin{figure*}
\centering
\tiny
$ 1. \left\langle \left[ {\begin{array}{cc}
\vspace{0.15cm}
1+j \\
\vspace{0.15cm}
1+j \\
2j \\
\end{array} } \right]\right\rangle
=\left\langle \left[ {\begin{array}{cc}
\vspace{0.15cm}
-1-j \\
\vspace{0.15cm}
-1-j \\
-2j \\
\end{array} } \right]\right\rangle
=\left\langle \left[ {\begin{array}{cc}
\vspace{0.15cm}
-1+j \\
\vspace{0.15cm}
-1+j \\
-2 \\
\end{array} } \right]\right\rangle 
=\left\langle \left[ {\begin{array}{cc}
\vspace{0.15cm}
1-j \\
\vspace{0.15cm}
1-j \\
2 \\
\end{array} } \right]\right\rangle $ ~~~~
$ 2. \left\langle \left[ {\begin{array}{cc}
\vspace{0.15cm}
1+j \\
\vspace{0.15cm}
1+j \\
-2j \\
\end{array} } \right]\right\rangle
=\left\langle \left[ {\begin{array}{cc}
\vspace{0.15cm}
-1-j \\
\vspace{0.15cm}
-1-j \\
2j \\
\end{array} } \right]\right\rangle
=\left\langle \left[ {\begin{array}{cc}
\vspace{0.15cm}
-1+j \\
\vspace{0.15cm}
-1+j \\
2 \\
\end{array} } \right]\right\rangle 
=\left\langle \left[ {\begin{array}{cc}
\vspace{0.15cm}
1-j \\
\vspace{0.15cm}
1-j \\
-2 \\
\end{array} } \right]\right\rangle $ \\
$ 3. \left\langle \left[ {\begin{array}{cc}
\vspace{0.15cm}
1+j \\
\vspace{0.15cm}
1+j \\
2 \\
\end{array} } \right]\right\rangle
=\left\langle \left[ {\begin{array}{cc}
\vspace{0.15cm}
-1-j \\
\vspace{0.15cm}
-1-j \\
-2 \\
\end{array} } \right]\right\rangle
=\left\langle \left[ {\begin{array}{cc}
\vspace{0.15cm}
-1+j \\
\vspace{0.15cm}
-1+j \\
2j \\
\end{array} } \right]\right\rangle 
=\left\langle \left[ {\begin{array}{cc}
\vspace{0.15cm}
1-j \\
\vspace{0.15cm}
1-j \\
-2j \\
\end{array} } \right]\right\rangle $ ~~~~
$ 4. \left\langle \left[ {\begin{array}{cc}
\vspace{0.15cm}
1+j \\
\vspace{0.15cm}
1+j \\
-2 \\
\end{array} } \right]\right\rangle
=\left\langle \left[ {\begin{array}{cc}
\vspace{0.15cm}
-1-j \\
\vspace{0.15cm}
-1-j \\
2 \\
\end{array} } \right]\right\rangle
=\left\langle \left[ {\begin{array}{cc}
\vspace{0.15cm}
-1+j \\
\vspace{0.15cm}
-1+j \\
-2j \\
\end{array} } \right]\right\rangle 
=\left\langle \left[ {\begin{array}{cc}
\vspace{0.15cm}
1-j \\
\vspace{0.15cm}
1-j \\
2j \\
\end{array} } \right]\right\rangle $ \\
$ 5. \left\langle \left[ {\begin{array}{cc}
\vspace{0.15cm}
1+j \\
\vspace{0.15cm}
-1-j \\
2j \\
\end{array} } \right]\right\rangle
=\left\langle \left[ {\begin{array}{cc}
\vspace{0.15cm}
-1-j \\
\vspace{0.15cm}
1+j \\
-2j \\
\end{array} } \right]\right\rangle
=\left\langle \left[ {\begin{array}{cc}
\vspace{0.15cm}
-1+j \\
\vspace{0.15cm}
1-j \\
-2 \\
\end{array} } \right]\right\rangle 
=\left\langle \left[ {\begin{array}{cc}
\vspace{0.15cm}
1-j \\
\vspace{0.15cm}
-1+j \\
2 \\
\end{array} } \right]\right\rangle $ ~~~~
$ 6. \left\langle \left[ {\begin{array}{cc}
\vspace{0.15cm}
1+j \\
\vspace{0.15cm}
-1-j \\
-2j \\
\end{array} } \right]\right\rangle
=\left\langle \left[ {\begin{array}{cc}
\vspace{0.15cm}
-1-j \\
\vspace{0.15cm}
1+j \\
2j \\
\end{array} } \right]\right\rangle
=\left\langle \left[ {\begin{array}{cc}
\vspace{0.15cm}
-1+j \\
\vspace{0.15cm}
1-j \\
2 \\
\end{array} } \right]\right\rangle 
=\left\langle \left[ {\begin{array}{cc}
\vspace{0.15cm}
1-j \\
\vspace{0.15cm}
-1+j \\
-2 \\
\end{array} } \right]\right\rangle $ \\
$ 7. \left\langle \left[ {\begin{array}{cc}
\vspace{0.15cm}
1+j\\
\vspace{0.15cm}
-1-j \\
2 \\
\end{array} } \right]\right\rangle
=\left\langle \left[ {\begin{array}{cc}
\vspace{0.15cm}
-1-j \\
\vspace{0.15cm}
1+j \\
-2 \\
\end{array} } \right]\right\rangle
=\left\langle \left[ {\begin{array}{cc}
\vspace{0.15cm}
-1+j \\
\vspace{0.15cm}
1-j \\
2j \\
\end{array} } \right]\right\rangle 
=\left\langle \left[ {\begin{array}{cc}
\vspace{0.15cm}
1-j \\
\vspace{0.15cm}
-1+j \\
-2j \\
\end{array} } \right]\right\rangle $ ~~~
$ 8. \left\langle \left[ {\begin{array}{cc}
\vspace{0.15cm}
1+j \\
\vspace{0.15cm}
-1-j \\
-2 \\
\end{array} } \right]\right\rangle
=\left\langle \left[ {\begin{array}{cc}
\vspace{0.15cm}
-1-j \\
\vspace{0.15cm}
1+j \\
2 \\
\end{array} } \right]\right\rangle
=\left\langle \left[ {\begin{array}{cc}
\vspace{0.15cm}
-1+j \\
\vspace{0.15cm}
1-j \\
-2j \\
\end{array} } \right]\right\rangle 
=\left\langle \left[ {\begin{array}{cc}
\vspace{0.15cm}
1-j \\
\vspace{0.15cm}
-1+j \\
2j \\
\end{array} } \right]\right\rangle. $\\
$ 9. \left\langle \left[ {\begin{array}{cc}
\vspace{0.15cm}
1+j \\
\vspace{0.15cm}
1-j \\
2j \\
\end{array} } \right]\right\rangle
=\left\langle \left[ {\begin{array}{cc}
\vspace{0.15cm}
-1-j \\
\vspace{0.15cm}
-1+j \\
-2j \\
\end{array} } \right]\right\rangle
=\left\langle \left[ {\begin{array}{cc}
\vspace{0.15cm}
-1+j \\
\vspace{0.15cm}
1+j \\
-2 \\
\end{array} } \right]\right\rangle 
=\left\langle \left[ {\begin{array}{cc}
\vspace{0.15cm}
1-j \\
\vspace{0.15cm}
-1-j \\
2 \\
\end{array} } \right]\right\rangle $ ~~~~
$ 10. \left\langle \left[ {\begin{array}{cc}
\vspace{0.15cm}
1+j \\
\vspace{0.15cm}
1-j \\
-2j \\
\end{array} } \right]\right\rangle
=\left\langle \left[ {\begin{array}{cc}
\vspace{0.15cm}
-1-j \\
\vspace{0.15cm}
-1+j \\
2j \\
\end{array} } \right]\right\rangle
=\left\langle \left[ {\begin{array}{cc}
\vspace{0.15cm}
-1+j \\
\vspace{0.15cm}
1+j \\
2 \\
\end{array} } \right]\right\rangle 
=\left\langle \left[ {\begin{array}{cc}
\vspace{0.15cm}
1-j \\
\vspace{0.15cm}
-1-j \\
-2 \\
\end{array} } \right]\right\rangle $ \\
$ 11. \left\langle \left[ {\begin{array}{cc}
\vspace{0.15cm}
1+j \\
\vspace{0.15cm}
1-j \\
2 \\
\end{array} } \right]\right\rangle
=\left\langle \left[ {\begin{array}{cc}
\vspace{0.15cm}
-1-j \\
\vspace{0.15cm}
-1+j \\
-2 \\
\end{array} } \right]\right\rangle
=\left\langle \left[ {\begin{array}{cc}
\vspace{0.15cm}
-1+j \\
\vspace{0.15cm}
1+j \\
2j \\
\end{array} } \right]\right\rangle 
=\left\langle \left[ {\begin{array}{cc}
\vspace{0.15cm}
1-j \\
\vspace{0.15cm}
-1-j \\
-2j \\
\end{array} } \right]\right\rangle $ ~~~~
$ 12. \left\langle \left[ {\begin{array}{cc}
\vspace{0.15cm}
1+j \\
\vspace{0.15cm}
1-j \\
-2 \\
\end{array} } \right]\right\rangle
=\left\langle \left[ {\begin{array}{cc}
\vspace{0.15cm}
-1-j \\
\vspace{0.15cm}
-1+j \\
2 \\
\end{array} } \right]\right\rangle
=\left\langle \left[ {\begin{array}{cc}
\vspace{0.15cm}
-1+j \\
\vspace{0.15cm}
1+j \\
-2j \\
\end{array} } \right]\right\rangle 
=\left\langle \left[ {\begin{array}{cc}
\vspace{0.15cm}
1-j \\
\vspace{0.15cm}
-1-j \\
2j \\
\end{array} } \right]\right\rangle $ \\
$ 13. \left\langle \left[ {\begin{array}{cc}
\vspace{0.15cm}
1+j \\
\vspace{0.15cm}
-1+j \\
2j \\
\end{array} } \right]\right\rangle
=\left\langle \left[ {\begin{array}{cc}
\vspace{0.15cm}
-1-j \\
\vspace{0.15cm}
1-j \\
-2j \\
\end{array} } \right]\right\rangle
=\left\langle \left[ {\begin{array}{cc}
\vspace{0.15cm}
-1+j \\
\vspace{0.15cm}
-1-j \\
-2 \\
\end{array} } \right]\right\rangle 
=\left\langle \left[ {\begin{array}{cc}
\vspace{0.15cm}
1-j \\
\vspace{0.15cm}
1+j \\
2 \\
\end{array} } \right]\right\rangle $ ~~~~
$ 14. \left\langle \left[ {\begin{array}{cc}
\vspace{0.15cm}
1+j \\
\vspace{0.15cm}
-1+j \\
-2j \\
\end{array} } \right]\right\rangle
=\left\langle \left[ {\begin{array}{cc}
\vspace{0.15cm}
-1-j \\
\vspace{0.15cm}
1-j \\
2j \\
\end{array} } \right]\right\rangle
=\left\langle \left[ {\begin{array}{cc}
\vspace{0.15cm}
-1+j \\
\vspace{0.15cm}
-1-j \\
2 \\
\end{array} } \right]\right\rangle 
=\left\langle \left[ {\begin{array}{cc}
\vspace{0.15cm}
1-j \\
\vspace{0.15cm}
1+j \\
-2 \\
\end{array} } \right]\right\rangle $ \\
$ 15. \left\langle \left[ {\begin{array}{cc}
\vspace{0.15cm}
1+j \\
\vspace{0.15cm}
-1+j \\
2 \\
\end{array} } \right]\right\rangle
=\left\langle \left[ {\begin{array}{cc}
\vspace{0.15cm}
-1-j \\
\vspace{0.15cm}
1-j \\
-2 \\
\end{array} } \right]\right\rangle
=\left\langle \left[ {\begin{array}{cc}
\vspace{0.15cm}
-1+j \\
\vspace{0.15cm}
-1-j \\
2j \\
\end{array} } \right]\right\rangle 
=\left\langle \left[ {\begin{array}{cc}
\vspace{0.15cm}
1-j \\
\vspace{0.15cm}
1+j \\
-2j \\
\end{array} } \right]\right\rangle $ ~~~
$ 16. \left\langle \left[ {\begin{array}{cc}
\vspace{0.15cm}
1+j \\
\vspace{0.15cm}
-1+j \\
-2 \\
\end{array} } \right]\right\rangle
=\left\langle \left[ {\begin{array}{cc}
\vspace{0.15cm}
-1-j \\
\vspace{0.15cm}
1-j \\
2 \\
\end{array} } \right]\right\rangle
=\left\langle \left[ {\begin{array}{cc}
\vspace{0.15cm}
-1+j \\
\vspace{0.15cm}
-1-j \\
-2j \\
\end{array} } \right]\right\rangle 
=\left\langle \left[ {\begin{array}{cc}
\vspace{0.15cm}
1-j \\
\vspace{0.15cm}
1+j \\
2j \\
\end{array} } \right]\right\rangle. $
\label{fig:sfssubcase2}
\caption{Null Spaces of the Singular Fades Spaces for the case $x_A - x'_A, ~ x_B-x'_B \in \mathcal{D}_1$ and $ x_C-x'_C \in \mathcal{D}_2$.}
\end{figure*}

\begin{figure*}
\centering
\tiny
$ 1. \left\langle \left[ {\begin{array}{cc}
\vspace{0.15cm}
1+j \\
\vspace{0.15cm}
1+j \\
1+j \\
\end{array} } \right]\right\rangle
=\left\langle \left[ {\begin{array}{cc}
\vspace{0.15cm}
-1-j \\
\vspace{0.15cm}
-1-j \\
-1-j \\
\end{array} } \right]\right\rangle
=\left\langle \left[ {\begin{array}{cc}
\vspace{0.15cm}
-1+j \\
\vspace{0.15cm}
-1+j \\
-1+j \\
\end{array} } \right]\right\rangle 
=\left\langle \left[ {\begin{array}{cc}
\vspace{0.15cm}
1-j \\
\vspace{0.15cm}
1-j \\
1-j \\
\end{array} } \right]\right\rangle
=\left\langle \left[ {\begin{array}{cc}
\vspace{0.15cm}
2j \\
\vspace{0.15cm}
2j \\
2j \\
\end{array} } \right]\right\rangle
=\left\langle \left[ {\begin{array}{cc}
\vspace{0.15cm}
-2j \\
\vspace{0.15cm}
-2j \\
-2j \\
\end{array} } \right]\right\rangle
=\left\langle \left[ {\begin{array}{cc}
\vspace{0.15cm}
2 \\
\vspace{0.15cm}
2 \\
2 \\
\end{array} } \right]\right\rangle 
=\left\langle \left[ {\begin{array}{cc}
\vspace{0.15cm}
-2 \\
\vspace{0.15cm}
-2 \\
-2 \\
\end{array} } \right]\right\rangle $\\ 
$ 2. \left\langle \left[ {\begin{array}{cc}
\vspace{0.15cm}
1+j \\
\vspace{0.15cm}
1+j \\
-1-j \\
\end{array} } \right]\right\rangle
=\left\langle \left[ {\begin{array}{cc}
\vspace{0.15cm}
-1-j \\
\vspace{0.15cm}
-1-j \\
1+j \\
\end{array} } \right]\right\rangle
=\left\langle \left[ {\begin{array}{cc}
\vspace{0.15cm}
-1+j \\
\vspace{0.15cm}
-1+j \\
1-j \\
\end{array} } \right]\right\rangle 
=\left\langle \left[ {\begin{array}{cc}
\vspace{0.15cm}
1-j \\
\vspace{0.15cm}
1-j \\
-1+j \\
\end{array} } \right]\right\rangle 
=\left\langle \left[ {\begin{array}{cc}
\vspace{0.15cm}
2j \\
\vspace{0.15cm}
2j \\
-2j \\
\end{array} } \right]\right\rangle
=\left\langle \left[ {\begin{array}{cc}
\vspace{0.15cm}
-2j \\
\vspace{0.15cm}
-2j \\
2j \\
\end{array} } \right]\right\rangle
=\left\langle \left[ {\begin{array}{cc}
\vspace{0.15cm}
2 \\
\vspace{0.15cm}
2 \\
-2 \\
\end{array} } \right]\right\rangle 
=\left\langle \left[ {\begin{array}{cc}
\vspace{0.15cm}
-2 \\
\vspace{0.15cm}
-2 \\
2 \\
\end{array} } \right]\right\rangle $\\ 
$ 3. \left\langle \left[ {\begin{array}{cc}
\vspace{0.15cm}
1+j \\
\vspace{0.15cm}
1+j \\
1-j \\
\end{array} } \right]\right\rangle
=\left\langle \left[ {\begin{array}{cc}
\vspace{0.15cm}
-1-j \\
\vspace{0.15cm}
-1-j \\
-1+j \\
\end{array} } \right]\right\rangle
=\left\langle \left[ {\begin{array}{cc}
\vspace{0.15cm}
-1+j \\
\vspace{0.15cm}
-1+j \\
1+j \\
\end{array} } \right]\right\rangle 
=\left\langle \left[ {\begin{array}{cc}
\vspace{0.15cm}
1-j \\
\vspace{0.15cm}
1-j \\
-1-j \\
\end{array} } \right]\right\rangle 
=\left\langle \left[ {\begin{array}{cc}
\vspace{0.15cm}
2j \\
\vspace{0.15cm}
2j \\
2 \\
\end{array} } \right]\right\rangle
=\left\langle \left[ {\begin{array}{cc}
\vspace{0.15cm}
-2j \\
\vspace{0.15cm}
-2j \\
-2 \\
\end{array} } \right]\right\rangle
=\left\langle \left[ {\begin{array}{cc}
\vspace{0.15cm}
2 \\
\vspace{0.15cm}
2 \\
-2j \\
\end{array} } \right]\right\rangle 
=\left\langle \left[ {\begin{array}{cc}
\vspace{0.15cm}
-2 \\
\vspace{0.15cm}
-2 \\
2j \\
\end{array} } \right]\right\rangle $\\ 
$ 4. \left\langle \left[ {\begin{array}{cc}
\vspace{0.15cm}
1+j \\
\vspace{0.15cm}
1+j \\
-1+j \\
\end{array} } \right]\right\rangle
=\left\langle \left[ {\begin{array}{cc}
\vspace{0.15cm}
-1-j \\
\vspace{0.15cm}
-1-j \\
1-j \\
\end{array} } \right]\right\rangle
=\left\langle \left[ {\begin{array}{cc}
\vspace{0.15cm}
-1+j \\
\vspace{0.15cm}
-1+j \\
-1-j \\
\end{array} } \right]\right\rangle 
=\left\langle \left[ {\begin{array}{cc}
\vspace{0.15cm}
1-j \\
\vspace{0.15cm}
1-j \\
1+j \\
\end{array} } \right]\right\rangle 
=\left\langle \left[ {\begin{array}{cc}
\vspace{0.15cm}
2j \\
\vspace{0.15cm}
2j \\
-2 \\
\end{array} } \right]\right\rangle
=\left\langle \left[ {\begin{array}{cc}
\vspace{0.15cm}
-2j \\
\vspace{0.15cm}
-2j \\
2 \\
\end{array} } \right]\right\rangle
=\left\langle \left[ {\begin{array}{cc}
\vspace{0.15cm}
2 \\
\vspace{0.15cm}
2 \\
2j \\
\end{array} } \right]\right\rangle 
=\left\langle \left[ {\begin{array}{cc}
\vspace{0.15cm}
-2 \\
\vspace{0.15cm}
-2 \\
-2j \\
\end{array} } \right]\right\rangle $\\ 
$ 5. \left\langle \left[ {\begin{array}{cc}
\vspace{0.15cm}
1+j \\
\vspace{0.15cm}
-1-j \\
1+j \\
\end{array} } \right]\right\rangle
=\left\langle \left[ {\begin{array}{cc}
\vspace{0.15cm}
-1-j \\
\vspace{0.15cm}
1+j \\
-1+j \\
\end{array} } \right]\right\rangle
=\left\langle \left[ {\begin{array}{cc}
\vspace{0.15cm}
-1+j \\
\vspace{0.15cm}
1-j \\
-1+j \\
\end{array} } \right]\right\rangle 
=\left\langle \left[ {\begin{array}{cc}
\vspace{0.15cm}
1-j \\
\vspace{0.15cm}
-1+j \\
1-j \\
\end{array} } \right]\right\rangle 
=\left\langle \left[ {\begin{array}{cc}
\vspace{0.15cm}
2j \\
\vspace{0.15cm}
-2j \\
2j \\
\end{array} } \right]\right\rangle
=\left\langle \left[ {\begin{array}{cc}
\vspace{0.15cm}
-2j \\
\vspace{0.15cm}
2j \\
-2j \\
\end{array} } \right]\right\rangle
=\left\langle \left[ {\begin{array}{cc}
\vspace{0.15cm}
2 \\
\vspace{0.15cm}
-2 \\
2 \\
\end{array} } \right]\right\rangle 
=\left\langle \left[ {\begin{array}{cc}
\vspace{0.15cm}
-2 \\
\vspace{0.15cm}
2 \\
-2 \\
\end{array} } \right]\right\rangle $\\ 
$ 6. \left\langle \left[ {\begin{array}{cc}
\vspace{0.15cm}
1+j \\
\vspace{0.15cm}
-1-j \\
-1-j \\
\end{array} } \right]\right\rangle
=\left\langle \left[ {\begin{array}{cc}
\vspace{0.15cm}
-1-j \\
\vspace{0.15cm}
1+j \\
1+j \\
\end{array} } \right]\right\rangle
=\left\langle \left[ {\begin{array}{cc}
\vspace{0.15cm}
-1+j \\
\vspace{0.15cm}
1-j \\
1-j \\
\end{array} } \right]\right\rangle 
=\left\langle \left[ {\begin{array}{cc}
\vspace{0.15cm}
1-j \\
\vspace{0.15cm}
-1+j \\
-1+j \\
\end{array} } \right]\right\rangle 
=\left\langle \left[ {\begin{array}{cc}
\vspace{0.15cm}
2j \\
\vspace{0.15cm}
-2j \\
-2j \\
\end{array} } \right]\right\rangle
=\left\langle \left[ {\begin{array}{cc}
\vspace{0.15cm}
-2j \\
\vspace{0.15cm}
2j \\
2j \\
\end{array} } \right]\right\rangle
=\left\langle \left[ {\begin{array}{cc}
\vspace{0.15cm}
2 \\
\vspace{0.15cm}
-2 \\
-2 \\
\end{array} } \right]\right\rangle 
=\left\langle \left[ {\begin{array}{cc}
\vspace{0.15cm}
-2 \\
\vspace{0.15cm}
2 \\
2 \\
\end{array} } \right]\right\rangle $\\ 
$ 7. \left\langle \left[ {\begin{array}{cc}
\vspace{0.15cm}
1+j\\
\vspace{0.15cm}
-1-j \\
-1+j \\
\end{array} } \right]\right\rangle
=\left\langle \left[ {\begin{array}{cc}
\vspace{0.15cm}
-1-j \\
\vspace{0.15cm}
1+j \\
1-j \\
\end{array} } \right]\right\rangle
=\left\langle \left[ {\begin{array}{cc}
\vspace{0.15cm}
-1+j \\
\vspace{0.15cm}
1-j \\
-1-j \\
\end{array} } \right]\right\rangle 
=\left\langle \left[ {\begin{array}{cc}
\vspace{0.15cm}
1-j \\
\vspace{0.15cm}
-1+j \\
1+j \\
\end{array} } \right]\right\rangle 
=\left\langle \left[ {\begin{array}{cc}
\vspace{0.15cm}
2j \\
\vspace{0.15cm}
-2j \\
-2 \\
\end{array} } \right]\right\rangle
=\left\langle \left[ {\begin{array}{cc}
\vspace{0.15cm}
-2j \\
\vspace{0.15cm}
2j \\
2 \\
\end{array} } \right]\right\rangle
=\left\langle \left[ {\begin{array}{cc}
\vspace{0.15cm}
2 \\
\vspace{0.15cm}
-2 \\
2j \\
\end{array} } \right]\right\rangle 
=\left\langle \left[ {\begin{array}{cc}
\vspace{0.15cm}
-2 \\
\vspace{0.15cm}
2 \\
-2j \\
\end{array} } \right]\right\rangle $\\ 
$ 8. \left\langle \left[ {\begin{array}{cc}
\vspace{0.15cm}
1+j \\
\vspace{0.15cm}
-1-j \\
1-j \\
\end{array} } \right]\right\rangle
=\left\langle \left[ {\begin{array}{cc}
\vspace{0.15cm}
-1-j \\
\vspace{0.15cm}
1+j \\
-1+j \\
\end{array} } \right]\right\rangle
=\left\langle \left[ {\begin{array}{cc}
\vspace{0.15cm}
-1+j \\
\vspace{0.15cm}
1-j \\
1+j \\
\end{array} } \right]\right\rangle 
=\left\langle \left[ {\begin{array}{cc}
\vspace{0.15cm}
1-j \\
\vspace{0.15cm}
-1+j \\
-1-j \\
\end{array} } \right]\right\rangle
=\left\langle \left[ {\begin{array}{cc}
\vspace{0.15cm}
2j \\
\vspace{0.15cm}
-2j \\
2 \\
\end{array} } \right]\right\rangle
=\left\langle \left[ {\begin{array}{cc}
\vspace{0.15cm}
-2j \\
\vspace{0.15cm}
2j \\
-2 \\
\end{array} } \right]\right\rangle
=\left\langle \left[ {\begin{array}{cc}
\vspace{0.15cm}
2 \\
\vspace{0.15cm}
-2 \\
-2j \\
\end{array} } \right]\right\rangle 
=\left\langle \left[ {\begin{array}{cc}
\vspace{0.15cm}
-2 \\
\vspace{0.15cm}
2 \\
2j \\
\end{array} } \right]\right\rangle $\\ 
$ 9. \left\langle \left[ {\begin{array}{cc}
\vspace{0.15cm}
1+j \\
\vspace{0.15cm}
1-j \\
1+j \\
\end{array} } \right]\right\rangle
=\left\langle \left[ {\begin{array}{cc}
\vspace{0.15cm}
-1-j \\
\vspace{0.15cm}
-1+j \\
-1-j \\
\end{array} } \right]\right\rangle
=\left\langle \left[ {\begin{array}{cc}
\vspace{0.15cm}
-1+j \\
\vspace{0.15cm}
1+j \\
-1+j \\
\end{array} } \right]\right\rangle 
=\left\langle \left[ {\begin{array}{cc}
\vspace{0.15cm}
1-j \\
\vspace{0.15cm}
-1-j \\
1-j \\
\end{array} } \right]\right\rangle 
=\left\langle \left[ {\begin{array}{cc}
\vspace{0.15cm}
2j \\
\vspace{0.15cm}
2 \\
2j \\
\end{array} } \right]\right\rangle
=\left\langle \left[ {\begin{array}{cc}
\vspace{0.15cm}
-2j \\
\vspace{0.15cm}
-2 \\
-2j \\
\end{array} } \right]\right\rangle
=\left\langle \left[ {\begin{array}{cc}
\vspace{0.15cm}
2 \\
\vspace{0.15cm}
-2j \\
2 \\
\end{array} } \right]\right\rangle 
=\left\langle \left[ {\begin{array}{cc}
\vspace{0.15cm}
-2 \\
\vspace{0.15cm}
-2 \\
2j \\
\end{array} } \right]\right\rangle $\\ 
$ 10. \left\langle \left[ {\begin{array}{cc}
\vspace{0.15cm}
1+j \\
\vspace{0.15cm}
1-j \\
-1-j \\
\end{array} } \right]\right\rangle
=\left\langle \left[ {\begin{array}{cc}
\vspace{0.15cm}
-1-j \\
\vspace{0.15cm}
-1+j \\
1+j \\
\end{array} } \right]\right\rangle
=\left\langle \left[ {\begin{array}{cc}
\vspace{0.15cm}
-1+j \\
\vspace{0.15cm}
1+j \\
1-j \\
\end{array} } \right]\right\rangle 
=\left\langle \left[ {\begin{array}{cc}
\vspace{0.15cm}
1-j \\
\vspace{0.15cm}
-1-j \\
-1+j \\
\end{array} } \right]\right\rangle 
=\left\langle \left[ {\begin{array}{cc}
\vspace{0.15cm}
2j \\
\vspace{0.15cm}
2 \\
-2j \\
\end{array} } \right]\right\rangle
=\left\langle \left[ {\begin{array}{cc}
\vspace{0.15cm}
-2j \\
\vspace{0.15cm}
-2 \\
2j \\
\end{array} } \right]\right\rangle
=\left\langle \left[ {\begin{array}{cc}
\vspace{0.15cm}
2 \\
\vspace{0.15cm}
-2j \\
-2 \\
\end{array} } \right]\right\rangle 
=\left\langle \left[ {\begin{array}{cc}
\vspace{0.15cm}
-2 \\
\vspace{0.15cm}
2j \\
2 \\
\end{array} } \right]\right\rangle $\\ 
$ 11. \left\langle \left[ {\begin{array}{cc}
\vspace{0.15cm}
1+j \\
\vspace{0.15cm}
1-j \\
1-j \\
\end{array} } \right]\right\rangle
=\left\langle \left[ {\begin{array}{cc}
\vspace{0.15cm}
-1-j \\
\vspace{0.15cm}
-1+j \\
-1+j \\
\end{array} } \right]\right\rangle
=\left\langle \left[ {\begin{array}{cc}
\vspace{0.15cm}
-1+j \\
\vspace{0.15cm}
1+j \\
1+j \\
\end{array} } \right]\right\rangle 
=\left\langle \left[ {\begin{array}{cc}
\vspace{0.15cm}
1-j \\
\vspace{0.15cm}
-1-j \\
-1-j \\
\end{array} } \right]\right\rangle 
=\left\langle \left[ {\begin{array}{cc}
\vspace{0.15cm}
2j \\
\vspace{0.15cm}
2 \\
2 \\
\end{array} } \right]\right\rangle
=\left\langle \left[ {\begin{array}{cc}
\vspace{0.15cm}
-2j \\
\vspace{0.15cm}
-2 \\
-2 \\
\end{array} } \right]\right\rangle
=\left\langle \left[ {\begin{array}{cc}
\vspace{0.15cm}
2 \\
\vspace{0.15cm}
-2j \\
-2j \\
\end{array} } \right]\right\rangle 
=\left\langle \left[ {\begin{array}{cc}
\vspace{0.15cm}
-2 \\
\vspace{0.15cm}
2j \\
2j \\
\end{array} } \right]\right\rangle $\\ 
$ 12. \left\langle \left[ {\begin{array}{cc}
\vspace{0.15cm}
1+j \\
\vspace{0.15cm}
1-j \\
-1+j \\
\end{array} } \right]\right\rangle
=\left\langle \left[ {\begin{array}{cc}
\vspace{0.15cm}
-1-j \\
\vspace{0.15cm}
-1+j \\
1-j \\
\end{array} } \right]\right\rangle
=\left\langle \left[ {\begin{array}{cc}
\vspace{0.15cm}
-1+j \\
\vspace{0.15cm}
1+j \\
-1-j \\
\end{array} } \right]\right\rangle 
=\left\langle \left[ {\begin{array}{cc}
\vspace{0.15cm}
1-j \\
\vspace{0.15cm}
-1-j \\
1+j \\
\end{array} } \right]\right\rangle 
=\left\langle \left[ {\begin{array}{cc}
\vspace{0.15cm}
2j \\
\vspace{0.15cm}
2 \\
-2 \\
\end{array} } \right]\right\rangle
=\left\langle \left[ {\begin{array}{cc}
\vspace{0.15cm}
-2j \\
\vspace{0.15cm}
-2 \\
2 \\
\end{array} } \right]\right\rangle
=\left\langle \left[ {\begin{array}{cc}
\vspace{0.15cm}
2 \\
\vspace{0.15cm}
-2j \\
2j \\
\end{array} } \right]\right\rangle 
=\left\langle \left[ {\begin{array}{cc}
\vspace{0.15cm}
-2 \\
\vspace{0.15cm}
2j \\
-2j \\
\end{array} } \right]\right\rangle $\\ 
$ 13. \left\langle \left[ {\begin{array}{cc}
\vspace{0.15cm}
1+j \\
\vspace{0.15cm}
-1+j \\
1+j \\
\end{array} } \right]\right\rangle
=\left\langle \left[ {\begin{array}{cc}
\vspace{0.15cm}
-1-j \\
\vspace{0.15cm}
1-j \\
-1-j \\
\end{array} } \right]\right\rangle
=\left\langle \left[ {\begin{array}{cc}
\vspace{0.15cm}
-1+j \\
\vspace{0.15cm}
-1-j \\
-1+j \\
\end{array} } \right]\right\rangle 
=\left\langle \left[ {\begin{array}{cc}
\vspace{0.15cm}
1-j \\
\vspace{0.15cm}
1+j \\
1-j \\
\end{array} } \right]\right\rangle 
=\left\langle \left[ {\begin{array}{cc}
\vspace{0.15cm}
2j \\
\vspace{0.15cm}
-2 \\
2j \\
\end{array} } \right]\right\rangle
=\left\langle \left[ {\begin{array}{cc}
\vspace{0.15cm}
-2j \\
\vspace{0.15cm}
2 \\
-2j \\
\end{array} } \right]\right\rangle
=\left\langle \left[ {\begin{array}{cc}
\vspace{0.15cm}
2 \\
\vspace{0.15cm}
2j \\
2 \\
\end{array} } \right]\right\rangle 
=\left\langle \left[ {\begin{array}{cc}
\vspace{0.15cm}
-2 \\
\vspace{0.15cm}
-2j \\
-2 \\
\end{array} } \right]\right\rangle $\\ 
$ 14. \left\langle \left[ {\begin{array}{cc}
\vspace{0.15cm}
1+j \\
\vspace{0.15cm}
-1+j \\
-1-j \\
\end{array} } \right]\right\rangle
=\left\langle \left[ {\begin{array}{cc}
\vspace{0.15cm}
-1-j \\
\vspace{0.15cm}
1-j \\
1+j \\
\end{array} } \right]\right\rangle
=\left\langle \left[ {\begin{array}{cc}
\vspace{0.15cm}
-1+j \\
\vspace{0.15cm}
-1-j \\
1-j \\
\end{array} } \right]\right\rangle 
=\left\langle \left[ {\begin{array}{cc}
\vspace{0.15cm}
1-j \\
\vspace{0.15cm}
1+j \\
-1+j \\
\end{array} } \right]\right\rangle 
=\left\langle \left[ {\begin{array}{cc}
\vspace{0.15cm}
2j \\
\vspace{0.15cm}
-2 \\
-2j \\
\end{array} } \right]\right\rangle
=\left\langle \left[ {\begin{array}{cc}
\vspace{0.15cm}
-2j \\
\vspace{0.15cm}
2 \\
2j \\
\end{array} } \right]\right\rangle
=\left\langle \left[ {\begin{array}{cc}
\vspace{0.15cm}
2 \\
\vspace{0.15cm}
2j \\
-2 \\
\end{array} } \right]\right\rangle 
=\left\langle \left[ {\begin{array}{cc}
\vspace{0.15cm}
-2 \\
\vspace{0.15cm}
-2j \\
2 \\
\end{array} } \right]\right\rangle $\\ 
$ 15. \left\langle \left[ {\begin{array}{cc}
\vspace{0.15cm}
1+j \\
\vspace{0.15cm}
-1+j \\
1-j \\
\end{array} } \right]\right\rangle
=\left\langle \left[ {\begin{array}{cc}
\vspace{0.15cm}
-1-j \\
\vspace{0.15cm}
1-j \\
-1+j \\
\end{array} } \right]\right\rangle
=\left\langle \left[ {\begin{array}{cc}
\vspace{0.15cm}
-1+j \\
\vspace{0.15cm}
-1-j \\
1+j \\
\end{array} } \right]\right\rangle 
=\left\langle \left[ {\begin{array}{cc}
\vspace{0.15cm}
1-j \\
\vspace{0.15cm}
1+j \\
-1-j \\
\end{array} } \right]\right\rangle 
=\left\langle \left[ {\begin{array}{cc}
\vspace{0.15cm}
2j \\
\vspace{0.15cm}
-2 \\
2 \\
\end{array} } \right]\right\rangle
=\left\langle \left[ {\begin{array}{cc}
\vspace{0.15cm}
-2j \\
\vspace{0.15cm}
2 \\
-2 \\
\end{array} } \right]\right\rangle
=\left\langle \left[ {\begin{array}{cc}
\vspace{0.15cm}
2 \\
\vspace{0.15cm}
2j \\
-2j \\
\end{array} } \right]\right\rangle 
=\left\langle \left[ {\begin{array}{cc}
\vspace{0.15cm}
-2 \\
\vspace{0.15cm}
-2j \\
2j \\
\end{array} } \right]\right\rangle $\\ 
$ 16. \left\langle \left[ {\begin{array}{cc}
\vspace{0.15cm}
1+j \\
\vspace{0.15cm}
-1+j \\
-1+j \\
\end{array} } \right]\right\rangle
=\left\langle \left[ {\begin{array}{cc}
\vspace{0.15cm}
-1-j \\
\vspace{0.15cm}
1-j \\
1-j \\
\end{array} } \right]\right\rangle
=\left\langle \left[ {\begin{array}{cc}
\vspace{0.15cm}
-1+j \\
\vspace{0.15cm}
-1-j \\
-1-j \\
\end{array} } \right]\right\rangle 
=\left\langle \left[ {\begin{array}{cc}
\vspace{0.15cm}
1-j \\
\vspace{0.15cm}
1+j \\
1+j \\
\end{array} } \right]\right\rangle
=\left\langle \left[ {\begin{array}{cc}
\vspace{0.15cm}
2j \\
\vspace{0.15cm}
-2 \\
-2 \\
\end{array} } \right]\right\rangle
=\left\langle \left[ {\begin{array}{cc}
\vspace{0.15cm}
-2j \\
\vspace{0.15cm}
2 \\
2 \\
\end{array} } \right]\right\rangle
=\left\langle \left[ {\begin{array}{cc}
\vspace{0.15cm}
2 \\
\vspace{0.15cm}
2j \\
2j \\
\end{array} } \right]\right\rangle 
=\left\langle \left[ {\begin{array}{cc}
\vspace{0.15cm}
-2 \\
\vspace{0.15cm}
-2j\\
-2j \\
\end{array} } \right]\right\rangle $.
\label{fig:sfssubcase3}
\caption{Null Spaces of the Singular Fades Spaces for the case $x_A - x'_A, ~ x_B-x'_B \text{~and~} x_C-x'_C \in \mathcal{D}_1$.}
\end{figure*}

\textit{Subcase 2:} Two of $x_A- x'_A,~ x_B-x'_B \text{~and~} x_C-x'_C \in \mathcal{D}_1$. Without loss of generality, we assume that $x_A-x'_A, x_B-x'_B \in \mathcal{D}_1$ and $x_C-x'_C \in \mathcal{D}_2$. The singular fade subspace for the case is given by (\ref{scase3}) and again there are 64 possibilities for the vector $v''=\left[x_{A}-x'_A, ~ x_B-x'_B, ~ x_C-x'_C\right]^{t}$. For a given $v''$, possibilities of other 3 length vectors over $ \mathcal{D}_1 \cup \mathcal{D}_2$ that generate the same vector space over $\mathbb{C}$ are the $\left\{\pm1, \pm j \right\}$ scalar multiples of $v''$. As a result, the case $x_A-x'_A \in \mathcal{D}_1$ and $x_B-x'_B, ~ x_C-x'_C \in \mathcal{D}_2$ also leads to 16 singular fade subspaces as shown for this case in Figure (\ref{fig:sfssubcase2}). The same holds for the case when $x_B-x'_B \in \mathcal{D}_2$ and $x_A-x'_A, ~ x_C-x'_C \in \mathcal{D}_1$, or when $x_A-x'_A \in \mathcal{D}_2$ and $x_B-x'_B, ~ x_C-x'_C \in \mathcal{D}_1$ resulting in therefore 48 singular fade subspaces. \\\\
\textit{Subcase 3:} All of $x_A- x'_A,~ x_B-x'_B \text{~and~} x_C-x'_C \in \mathcal{D}_1$. There are 64 possibilities for the vector $\left[x_{A}-x'_A, ~ x_B-x'_B, ~ x_C-x'_C\right]^{t}$ over $\mathcal{D}_1$. For a given such vector, possibilities of other 3 length vectors over $ \mathcal{D}_1 \cup \mathcal{D}_2$ that generate the same vector space over $\mathbb{C}$ are the $\left\{\pm1, \pm j, \pm1\pm j \right\}$ scalar multiples of the vector. The case $x_A-x'_A, ~x_B-x'_B, ~ x_C-x'_C \in \mathcal{D}_1$ leads to a total of 16 singular fade subspaces as shown in Figure (\ref{fig:sfssubcase3}).\\

The three cases result in a total of 3+36+48+48+16=151 singular fade subspaces. We now discuss how these singular fade subspaces can be removed using Latin Cubes of Second order.

\section{Removing singular fade subspaces and Constrained Latin Cubes}

In the previous section, we classify the set of singular fade subspaces into three cases. We now cluster the possibilities of $\left(x_{A},x_{B},x_{C}\right)$ into a clustering using Latin Cubes. This clustering is represented by a constellation given by $\mathcal{S}'$, which is utilized by the relay node R in the BC phase. The objective is to minimize the size of this constellation used by R. \\

The clustering to be used at R, first constrains the possibilities of $\left(x_{A},x_{B},x_{C}\right)$ received at the MA phase, with the objective of removing the singular fade subspaces, and fills the entries of a $4 \times 4\times 4 $ array representing the map to be used at the relay using these constraints, and then completes the partially empty array obtained to form a Latin cube of second order. The mapping to be used at R can be obtained from the complete Latin cube keeping in mind the equivalence between the relay map with the Latin Cube of second order as shown in Section III. \\

In order to to obtain the constraints on the $4 \times 4 \times 4$ array representing the map at the relay node R during BC phase for a singular fade state, we utilize the vectors of differences, viz., $ \left[x_A-x'_A ~ x_B-x'_B,~ x_C-x'_C\right]$ contributing to that particular singular fade state. During MA phase for the three-way relaying scenario, nodes A, B and C transmit to the relay R. As shown in the previous section, there is a total of 151 singular fade subspaces. Let $(h_A, h_B, h_C)$ denote a point in one of the 151 singular fade subspaces. Then, there exist $\left(x_{A}, x_{B}, x_C\right), \left(x'_{A},x'_{B},x'_C\right) \in \mathcal{S}^3$ that satisfy $ h_A x_{A} + h_B x_{B} + h_C x_C= h_A x'_{A} + h_B x'_{B} + h_C x'_C$. In order to remove the singular fade state $(h_A, h_B, h_C)$, the pair $\left(x_{A}, x_{B}, x_C\right), \left(x'_{A},x'_{B},x'_C\right)$ must be kept in the same cluster in the clustering. For instance, if 
$$(h_A, h_B, h_C) \in \left\langle \left[ {\begin{array}{cc}
\vspace{0.15cm}
x_A- x'_A \\
\vspace{0.15cm}
x_B-x'_B\\
x_C-x'_C \\
\end{array} } \right]\right\rangle ^{\bot},  $$ then the pair $\left(x_{A}, x_{B}, x_C\right), \left(x'_{A},x'_{B},x'_C\right)$ must be kept in the same cluster in the clustering, i.e., the entry corresponding to $\left(x_{A}, x_{B}, x_C\right)$ in the $4 \times 4 \times 4$ array must be the same as the entry corresponding to $\left(x'_{A},x'_{B},x'_C\right)$. Similarly, every other such pair in $\mathcal{S}^3$ contributing to the same singular fade subspace must be kept in the same cluster. Apart from all such pairs in $\mathcal{S}^3$ being kept in the same cluster of the clustering, in order to remove this particular fade state, there are no other constraints. Upon filling up the $4 \times 4 \times 4$ array with these constraints, the entire Latin Cube does not get filled up, but can be completed as we show later in this section. It is important to note that, this clustering cannot be utilized to remove the singular fade subspaces of Case 1 and Case 2 of the previous section, as shown in the following lemma.\\

\begin{lemma}
The clustering map used at the relay node R cannot remove a singular fade state corresponding to the Case 1 and Case 2 of the previous section and simultaneously satisfy the mutually exclusive law.
\end{lemma}
\begin{proof}
Let $\mathcal{S}'=\left\langle \left[ {\begin{array}{cc}
\vspace{0.15cm}
x_A- x'_A \\
\vspace{0.15cm}
x_B-x'_B\\
x_C-x'_C \\
\end{array} } \right]\right\rangle ^{\bot} $ be a singular fade state for Case 1, and without loss of generality, assume that $x_B=x'_B \text{~and~} x_C=x'_C$. Then, in order to remove $\mathcal{S}'$, $(x_A, x_B, x_C)$ and $(x'_A, x_B, x_C)$ must be kept in the same cluster. This is impossible, since this clearly violates the mutually exclusive law as if the pair is placed in the same cluster, the users B and C will not be able to distinguish between the messages $x_A$ and $x'_A$ sent by the user A.\\

Let $\mathcal{S}'$ be a singular fade state for Case 2, and without loss of generality, assume that $x_C=x'_C$. Then, in order to remove $\mathcal{S}'$, $(x_A, x_B, x_C)$ and $(x'_A, x'_B, x_C)$ must be kept in the same cluster. This also clearly violates the mutually exclusive law since if the pair is placed in the same cluster, the user C will not be able to distinguish between the messages $x_A,~x_B$ and messages $x'_A, ~x'_B$ sent by the users A and B.
\end{proof}
The singular fade subspaces given in Case 1 ans Case 2 of the previous section, whose harmful effects cannot be removed by a proper choice of the clustering are referred to as the \textit{non-removable singular fade subspaces}. The rest of the singular fade subspaces, given in Case 3 of the previous section, are referred as the \textit{removable singular fade subspaces}.

We now illustrate the removal of a Case 3 singular fade state with the help of examples.\\
\begin{example}
Let the singular fade subspace be 
{\tiny
$$\mathcal{S}'= \left\langle \left[ {\begin{array}{cc}
\vspace{0.15cm}
1+j \\
\vspace{0.15cm}
2j \\
-2j \\
\end{array} } \right]\right\rangle ^{\bot}
=\left\langle \left[ {\begin{array}{cc}
\vspace{0.15cm}
-1-j \\
\vspace{0.15cm}
-2j \\
2j \\
\end{array} } \right]\right\rangle ^{\bot}
=\left\langle \left[ {\begin{array}{cc}
\vspace{0.15cm}
-1+j \\
\vspace{0.15cm}
-2 \\
2 \\
\end{array} } \right]\right\rangle ^{\bot}
=\left\langle \left[ {\begin{array}{cc}
\vspace{0.15cm}
1-j \\
\vspace{0.15cm}
2 \\
-2 \\
\end{array} } \right]\right\rangle ^{\bot} $$
}

Consider the first vector $\left[ 1+j, ~ 2j, ~ -2j \right]^t$. Here, $1+j$ can be obtained as a difference of $x_A=1$ and $x'_A=-j$ or as a difference of $x_A=j$ and $x'_A=-1$; $2j$ can be obtained as a difference of $x_B=j$ and $x'_B=-j$; and $-2j$ can be obtained as a difference of $x_C=-j$ and $x'_C=j$. Thus, the entries corresponding to $(1,j,-j)$ and $(-j,-j,j)$ must be the same and the entries corresponding to $(j,j,-j)$ and $(-1,-j,j)$ must be the same in the $4 \times 4 \times 4$ array representing the clustering, i.e., entries $(0,1,3)$ and $(3,3,1)$ must be the same and entries $(1,1,3)$ and $(2,3,1)$ must be the same.\\

The second vector $\left[ -1-j,~ -2j, ~ 2j \right]^t$ where, $-1-j$ can be obtained as a difference of $x_A=-1$ and $x'_A=j$ or as a difference of $x_A=-j$ and $x'_A=1$; $-2j$ can be obtained as a difference of $x_B=-j$ and $x'_B=j$; and $2j$ can be obtained as a difference of $x_C=j$ and $x'_C=-j$. Thus, the entries corresponding to $(-1,-j,j)$ and $(j,j,-j)$ must be the same and the entries corresponding to $(-j,-j,j)$ and $(1,j,-j)$ must be the same in the $4 \times 4 \times 4$ array representing the clustering, i.e., entries $(2,3,1)$ and $(1,1,3)$ must be the same and entries $(3,3,1)$ and $(0,1,3)$ must be the same. Similarly, the constraints resulting from the third and forth vector ca be obtained. The constrained array is shown in Figure 8.\\

We choose to fill the constrained array to form a Latin cube of second order as shown in adjoining Algorithm 1, which simply fills the empty cells of the partially filled $4 \times 4 \times 4$ array with $\mathcal{L}_{i}, ~ i \geq 1$ in increasing order of $i$ keeping in mind that the resulting array must be a Latin cube of second order. The top-most and the left-most empty cell in the earliest file is filled at every iteration. The completed Latin cube is shown in Figure 9.\\
\begin{figure*}[tp]
{\tiny 
\centering
%\begin{tabular}{c}
%\vspace{0.2cm}
\begin{minipage}[b]{0.24\linewidth}
\centering
%\begin{tabular}{|c|p{0.05cm}|p{0.05cm}|p{0.05cm}|p{0.05cm}|}\hline
\begin{tabular}{|c|c|c|c|c|}\hline 
        $x_A=0$ & 0 & 1 & 2 & 3 \\\hline \hline
        0   &                    &                    &$\mathcal{L}_{3}$   &            \\\hline
        1   &                   &                     &                    &$\mathcal{L}_{2}$  \\\hline
        2   &                   &                     &                   &                    \\\hline
        3   &                   &                     &                   &                   \\\hline
\end{tabular}
\end{minipage} 
\begin{minipage}[b]{0.24\linewidth}
\centering
\begin{tabular}{|c|c|c|c|c|}\hline 
        $x_A=1$ & 0 & 1 & 2 & 3 \\\hline \hline
        0   &                  &                    &                   &                    \\\hline
        1   &                   &                   &                    &$\mathcal{L}_{1}$  \\\hline
        2   &$\mathcal{L}_{3} $  &                  &                     &                   \\\hline
        3   &                   &                   &                     &                   \\\hline
\end{tabular}
\end{minipage} 
\begin{minipage}[b]{0.22\linewidth}
\centering
\begin{tabular}{|c|c|c|c|c|}\hline 
        $x_A=2$ & 0 & 1 & 2 & 3 \\\hline \hline
        0   &                   &                    &                    &                    \\\hline
        1   &                   &                    &                    &                   \\\hline
        2   &$\mathcal{L}_{4} $ &                    &                   &                    \\\hline
        3   &                   &$\mathcal{L}_{1}$  &                    &                  \\\hline
\end{tabular}
\end{minipage} 
\begin{minipage}[b]{0.24\linewidth}
\vspace{0.2cm}
\centering
\begin{tabular}{|c|c|c|c|c|}\hline 
        $x_A=3$ & 0 & 1 & 2 & 3 \\\hline \hline
        0   &                    &                    &$\mathcal{L}_{4}$   &                     \\\hline
        1   &                     &                   &                    &                    \\\hline
        2   &                    &                   &                    &                    \\\hline
        3   &                     &$\mathcal{L}_{2}$  &                     &                    \\\hline
\end{tabular}
\end{minipage}
\vspace{0.2cm}
%\end{tabular}
\caption{Constraints for Example 1, with B's transmitted symbols along the rows and C's transmitted symbols along the columns}}
\end{figure*}

\end{example}

\begin{algorithm}
\SetLine
\linesnumbered
\KwIn{The constrained $4 \times 4\times 4$ array}
\KwOut{A Latin Cube of Second Order representing the clustering map at the relay}

Start with the constrained $4 \times 4\times 4$ array

Initialize all empty cells of $\mathcal{X}$ to 0

Let $\mathcal{Y}$ denote the $4 \times 16$ matrix obtained by concatenating $\mathcal{X}$ row-wise, and let $\mathcal{Z}$ denote the $16 \times 4$ matrix obtained by concatenating $\mathcal{X}$ column-wise

The $\left(i,j,k\right)^{th}$ element of $\mathcal{X}$ is the $i^{th}$ file, the $j^{th}$ row and the $k^{th}$ column cell.

\For{$1\leq i\leq 4 $}{

\For{$1\leq j\leq 4 $}{

\For{$1\leq k\leq 4 $}{

\If{cell $\left(i,j,k\right)$ of $\mathcal{X}$ is NULL}{

Initialize c=1

\eIf{$\mathcal{L}_{c}$ does not occur in the $i^{th}$ file of $\mathcal{X}$, the $j^{th}$ row of $\mathcal{Y}$ and the $k^{th}$ column of $\mathcal{Z}$}{
replace 0 at cell $\left(i,j,k\right)$ of $\mathcal{X}$ with $\mathcal{L}_{c}$\;
replace $\mathcal{Y}$ with the $4 \times 16$ matrix obtained by concatenating $\mathcal{X}$ row-wise, and $\mathcal{Z}$ by the $16 \times 4$ matrix obtained by concatenating $\mathcal{X}$ column-wise\;
}{
c=c+1\; 
}
}
}
}
}
%
%\For{$c' \geq 1$}{
%
%\If{$\mathcal{L}_{c'}$ occurs twice at positions $\left(i_{1}, j_{1}, k_{1}\right)$ and $\left(i_{2}, j_{2}, k_{2}\right)$ in $\mathcal{X}$}{
%
%\For{$d \geq 1$}{
%
%\If{$\mathcal{L}_{d}$ occurs thrice in $\mathcal{X}$}{
%
%\If{$\mathcal{L}_{d}$ does not occur in the ${i_{1}}^{th}$ file of $\mathcal{X}$, the ${j_{1}}^{th}$ row of $\mathcal{Y}$ and the ${k_{1}}^{th}$ column of $\mathcal{Z}$}{
%replace $\mathcal{L}_{c'}$ at cell $\left(i_{1}, j_{1}, k_{1}\right)$ of $\mathcal{A}$ with $\mathcal{L}_{d}$\;
%replace $\mathcal{Y}$ with the $4 \times 16$ matrix obtained by concatenating $\mathcal{X}$ row-wise, and $\mathcal{Z}$ by the $16 \times 4$ matrix obtained by concatenating $\mathcal{X}$ column-wise\;
%
%}
%}
%}
%
%repeat for position $\left(i_{2}, j_{2}, k_{2}\right)$\;
%
%}
%}
\caption{Obtaining the Latin Cube of Second Order from the constrained $4 \times 4 \times 4 $ array}

\end{algorithm}

\begin{figure*}[tp]
{\tiny
\centering
%\begin{tabular}{c}
%\vspace{0.2cm}
\begin{minipage}[b]{0.24\linewidth}
%\centering
\begin{tabular}{|c|p{0.19cm}|p{0.19cm}|p{0.19cm}|p{0.19cm}|}\hline        $x_A=0$ & 0 & 1 & 2 & 3 \\\hline \hline
        0   & $\mathcal{L}_{1}$ &  $\mathcal{L}_{5}$ &$\pmb{\mathcal{L}}_{\pmb{3}}$   & $\mathcal{L}_{6}$  \\\hline
        1   & $\mathcal{L}_{7}$ &  $\mathcal{L}_{4}$  & $\mathcal{L}_{8}$   &$\pmb{\mathcal{L}}_{\pmb{2}}$  \\\hline
        2   & $\mathcal{L}_{9}$ &  $\mathcal{L}_{10}$  & $\mathcal{L}_{11}$ &  $\mathcal{L}_{12}$              \\\hline
        3   & $\mathcal{L}_{13}$ &   $\mathcal{L}_{14}$ & $\mathcal{L}_{15}$ &  $\mathcal{L}_{16}$             \\\hline
\end{tabular}
\end{minipage} 
\begin{minipage}[b]{0.24\linewidth}
%\centering
\begin{tabular}{|c|p{0.19cm}|p{0.19cm}|p{0.19cm}|p{0.19cm}|}\hline 
        $x_A=1$ & 0 & 1 & 2 & 3 \\\hline \hline
        0   & $\mathcal{L}_{2}$ & $\mathcal{L}_{7}$  &  $\mathcal{L}_{9}$ &  $\mathcal{L}_{8}$  \\\hline
        1   & $\mathcal{L}_{5}$  & $\mathcal{L}_{6}$  &  $\mathcal{L}_{10}$ &$\pmb{\mathcal{L}}_{\pmb{1}}$  \\\hline
        2   &$\pmb{\mathcal{L}}_{\pmb{3}} $  &  $\mathcal{L}_{13}$   &  $\mathcal{L}_{14}$   &  $\mathcal{L}_{15}$   \\\hline
        3   & $\mathcal{L}_{11}$  &  $\mathcal{L}_{12}$  &  $\mathcal{L}_{17}$   &  $\mathcal{L}_{4}$  \\\hline
\end{tabular}
\end{minipage} 
\begin{minipage}[b]{0.24\linewidth}
%\centering
\begin{tabular}{|c|p{0.19cm}|p{0.19cm}|p{0.19cm}|p{0.19cm}|}\hline         $x_A=2$ & 0 & 1 & 2 & 3 \\\hline \hline
        0   & $\mathcal{L}_{10}$ & $\mathcal{L}_{11}$ &  $\mathcal{L}_{12}$ &  $\mathcal{L}_{13}$  \\\hline
        1   & $\mathcal{L}_{14}$ & $\mathcal{L}_{3}$ & $\mathcal{L}_{16}$  &  $\mathcal{L}_{9}$ \\\hline
        2   &$\pmb{\mathcal{L}}_{\pmb{4}} $ & $\mathcal{L}_{8}$ &$\mathcal{L}_{2}$ &  $\mathcal{L}_{5}$   \\\hline
        3   & $\mathcal{L}_{6}$ &$\pmb{\mathcal{L}}_{\pmb{1}}$  & $\mathcal{L}_{7}$  & $\mathcal{L}_{18}$ \\\hline
\end{tabular}
\end{minipage} 
\begin{minipage}[b]{0.24\linewidth}
\vspace{0.2 cm}
%\centering
\begin{tabular}{|c|p{0.19cm}|p{0.19cm}|p{0.19cm}|p{0.19cm}|}\hline         $x_A=3$ & 0 & 1 & 2 & 3 \\\hline \hline
        0   & $\mathcal{L}_{15}$  &  $\mathcal{L}_{16}$ &$\pmb{\mathcal{L}}_{\pmb{4}}$   &  $\mathcal{L}_{14}$    \\\hline
        1   & $\mathcal{L}_{12}$ & $\mathcal{L}_{17}$  &  $\mathcal{L}_{13}$  & $\mathcal{L}_{11}$ \\\hline
        2   & $\mathcal{L}_{18}$  & $\mathcal{L}_{19}$ &  $\mathcal{L}_{1}$   &  $\mathcal{L}_{7}$   \\\hline
        3   &  $\mathcal{L}_{8}$   &$\pmb{\mathcal{L}}_{\pmb{2}}$  & $\mathcal{L}_{5}$    & $\mathcal{L}_{3}$   \\\hline
\end{tabular}
\end{minipage}
%\end{tabular}
\caption{Latin Cube for Example 1, with B's transmitted symbols along the rows and C's transmitted symbols along the columns}}
\end{figure*}

\begin{example}
Consider another example for which the singular fade subspace is given by
{\tiny
$$\mathcal{S}''= \left\langle \left[ {\begin{array}{cc}
\vspace{0.15cm}
1+j \\
\vspace{0.15cm}
1+j \\
-1-j \\
\end{array} } \right]\right\rangle ^{\bot}
=\left\langle \left[ {\begin{array}{cc}
\vspace{0.15cm}
-1-j \\
\vspace{0.15cm}
-1-j \\
1+j \\
\end{array} } \right]\right\rangle ^{\bot}
=\left\langle \left[ {\begin{array}{cc}
\vspace{0.15cm}
-1+j \\
\vspace{0.15cm}
-1+j \\
1-j \\
\end{array} } \right]\right\rangle ^{\bot}
=\left\langle \left[ {\begin{array}{cc}
\vspace{0.15cm}
1-j \\
\vspace{0.15cm}
1-j \\
-1+j \\
\end{array} } \right]\right\rangle ^{\bot} $$
}
{\tiny
$$~~~~~~~= \left\langle \left[ {\begin{array}{cc}
\vspace{0.15cm}
2j \\
\vspace{0.15cm}
2j \\
-2j \\
\end{array} } \right]\right\rangle ^{\bot}
=\left\langle \left[ {\begin{array}{cc}
\vspace{0.15cm}
-2j \\
\vspace{0.15cm}
-2j \\
2j \\
\end{array} } \right]\right\rangle ^{\bot}
=\left\langle \left[ {\begin{array}{cc}
\vspace{0.15cm}
-2 \\
\vspace{0.15cm}
-2 \\
2 \\
\end{array} } \right]\right\rangle ^{\bot}
=\left\langle \left[ {\begin{array}{cc}
\vspace{0.15cm}
2 \\
\vspace{0.15cm}
2 \\
-2 \\
\end{array} } \right]\right\rangle ^{\bot}. $$
}

The first vector is $\left[ 1+j, ~ 1+j, ~ -1-j \right]$. Here, $1+j$ can be obtained as a difference of $x_A=1$ and $x'_A=-j$ or as a difference of $x_A=j$ and $x'_A=-1$; $-1-j$ can be obtained as a difference of $x_C=-1$ and $x'_C=j$ or as a difference of $x_C=-j$ and $x'_C=1$. Thus, the entries corresponding to {\footnotesize $\left\{(1,1,-1),(-j,-j,j)\right\}$, $\left\{(1,1,-j),(-j,-j,1)\right\}$,$\left\{(1,j,-1),(-j,-1,j)\right\}$,$\left\{(1,j,-j),(-j,-1,1)\right\}$, $\left\{(j,1,-1),(-1,-j,j)\right\}$,$\left\{(j,1,-j),(-1,-j,1)\right\}$,$\left\{(j,j,-1),(-1,-1,j)\right\}$, $\left\{(j,j,-j),(-1,-1,1)\right\}$ } must be the same in the $4 \times 4 \times 4$ array representing the clustering, i.e., entries $\left\{(0,0,2),(3,3,1)\right\}$, $\left\{(0,0,3),(3,3,0)\right\}$, $\left\{(0,1,2),(3,2,1)\right\}$, $\left\{(0,1,3),(3,2,0)\right\}$, $\left\{(1,0,2),(2,3,1)\right\}$, $\left\{(1,0,3),(2,3,0)\right\}$, $\left\{(1,1,2),(2,2,1)\right\}$, $\left\{(1,1,3),(2,2,0)\right\}$ must be the same. Similarly the other constraints can be obtained. The constrained array is shown in Figure \ref{fig:ex2fig1}, and the clustering is as shown in Figure \ref{fig:ex2fig2}.\\

\begin{figure*}[tp]
\centering
{\tiny
%\begin{tabular}{c}
%\vspace{0.2cm}
\begin{minipage}[b]{0.24\linewidth}
%\centering
\begin{tabular}{|c|c|c|c|c|}\hline
        $x_A=0$ & 0 & 1 & 2 & 3 \\\hline \hline
        0   &                   &  $\mathcal{L}_{5}$    &$\mathcal{L}_{1}$   &  $\mathcal{L}_{2}$   \\\hline
        1   &                   &                     & $\mathcal{L}_{3}$    &$\mathcal{L}_{4}$  \\\hline
        2   &                   &                     &                   &                    \\\hline
        3   &                   &  $\mathcal{L}_{6}$   & $\mathcal{L}_{7}$    &                   \\\hline
\end{tabular}
\end{minipage} 
\begin{minipage}[b]{0.24\linewidth}
\centering
\begin{tabular}{|c|c|c|c|c|}\hline
        $x_A=1$ & 0 & 1 & 2 & 3 \\\hline \hline
        0   &                  &                    &  $\mathcal{L}_{4}$   &  $\mathcal{L}_{3}$    \\\hline
        1   & $\mathcal{L}_{5}$    &                   &  $\mathcal{L}_{2}$    &$\mathcal{L}_{1}$  \\\hline
        2   &$\mathcal{L}_{6} $  &                  &                     &  $\mathcal{L}_{7}$    \\\hline
        3   &                   &                   &                     &                   \\\hline
\end{tabular}
\end{minipage} 
\begin{minipage}[b]{0.24\linewidth}
\centering
\begin{tabular}{|c|c|c|c|c|}\hline
        $x_A=2$ & 0 & 1 & 2 & 3 \\\hline \hline
        0   &                   &                    &                    &                    \\\hline
        1   & $\mathcal{L}_{7}$    &                    &                    &  $\mathcal{L}_{6}$  \\\hline
        2   &$\mathcal{L}_{1} $ & $\mathcal{L}_{2}$     &                   &  $\mathcal{L}_{5}$     \\\hline
        3   & $\mathcal{L}_{3}$  &$\mathcal{L}_{4}$  &                    &                  \\\hline
\end{tabular}
\end{minipage} 
\begin{minipage}[b]{0.24\linewidth}
\vspace{0.2cm}
\centering
\begin{tabular}{|c|c|c|c|c|}\hline
        $x_A=3$ & 0 & 1 & 2 & 3 \\\hline \hline
        0   &                    &  $\mathcal{L}_{7}$   &$\mathcal{L}_{6}$   &                     \\\hline
        1   &                     &                   &                    &                    \\\hline
        2   & $\mathcal{L}_{4}$    &  $\mathcal{L}_{3}$     &                    &                    \\\hline
        3   & $\mathcal{L}_{2}$     &$\mathcal{L}_{1}$  &  $\mathcal{L}_{5}$     &                    \\\hline
\end{tabular}
\end{minipage}
\vspace{0.2cm}
%\end{tabular}
\caption{Constraints for Example 1, with B's transmitted symbols along the rows and C's transmitted symbols along the columns}
\label{fig:ex2fig1}}
\end{figure*}

\begin{figure*}[tp]
\centering
{\tiny
%\begin{tabular}{c}
\begin{minipage}[b]{0.24\linewidth}
\centering
\begin{tabular}{|c|p{0.19cm}|p{0.19cm}|p{0.19cm}|p{0.19cm}|}\hline  
        $x_A=0$ & 0 & 1 & 2 & 3 \\\hline \hline
        0   & $\mathcal{L}_{8}$  &  $\pmb{\mathcal{L}}_{\pmb{5}}$    &$\pmb{\mathcal{L}}_{\pmb{1}}$   &  $\pmb{\mathcal{L}}_{\pmb{2}}$   \\\hline
        1   & $\mathcal{L}_{9}$  &  $\mathcal{L}_{10}$  & $\pmb{\mathcal{L}}_{\pmb{3}}$    &$\pmb{\mathcal{L}}_{\pmb{4}}$  \\\hline
        2   & $\mathcal{L}_{11}$ &   $\mathcal{L}_{12}$ & $\mathcal{L}_{13}$  & $\mathcal{L}_{14}$  \\\hline
        3   & $\mathcal{L}_{15}$ &  $\pmb{\mathcal{L}}_{\pmb{6}}$   & $\pmb{\mathcal{L}}_{\pmb{7}}$    &  $\mathcal{L}_{16}$   \\\hline
\end{tabular}
\end{minipage} 
\begin{minipage}[b]{0.24\linewidth}
\centering
\begin{tabular}{|c|p{0.19cm}|p{0.19cm}|p{0.19cm}|p{0.19cm}|}\hline          $x_A=1$ & 0 & 1 & 2 & 3 \\\hline \hline
        0   & $\mathcal{L}_{10}$  &  $\mathcal{L}_{9}$  &  $\pmb{\mathcal{L}}_{\pmb{4}}$   &  $\pmb{\mathcal{L}}_{\pmb{3}}$    \\\hline
        1   & $\pmb{\mathcal{L}}_{\pmb{5}}$    & $\mathcal{L}_{8}$   &  $\pmb{\mathcal{L}}_{\pmb{2}}$    &$\pmb{\mathcal{L}}_{\pmb{1}}$  \\\hline
        2   &$\pmb{\mathcal{L}}_{\pmb{6}} $  &  $\mathcal{L}_{15}$ &  $\mathcal{L}_{16}$  &  $\pmb{\mathcal{L}}_{\pmb{7}}$    \\\hline
        3   & $\mathcal{L}_{12}$  & $\mathcal{L}_{11}$  &  $\mathcal{L}_{14}$  &  $\mathcal{L}_{13}$  \\\hline
\end{tabular}
\end{minipage} 
\begin{minipage}[b]{0.24\linewidth}
\centering
\begin{tabular}{|c|p{0.19cm}|p{0.19cm}|p{0.19cm}|p{0.19cm}|}\hline          $x_A=2$ & 0 & 1 & 2 & 3 \\\hline \hline
        0   & $\mathcal{L}_{13}$  &  $\mathcal{L}_{14}$ & $\mathcal{L}_{11}$   &  $\mathcal{L}_{12}$   \\\hline
        1   & $\pmb{\mathcal{L}}_{\pmb{7}}$    &  $\mathcal{L}_{16}$  &  $\mathcal{L}_{15}$   &  $\pmb{\mathcal{L}}_{\pmb{6}}$  \\\hline
        2   &$\pmb{\mathcal{L}}_{\pmb{1}} $ & $\pmb{\mathcal{L}}_{\pmb{2}}$     &  $\mathcal{L}_{8}$  &  $\pmb{\mathcal{L}}_{\pmb{5}}$     \\\hline
        3   & $\pmb{\mathcal{L}}_{\pmb{3}}$  &$\pmb{\mathcal{L}}_{\pmb{4}}$  & $\mathcal{L}_{9}$   &   $\mathcal{L}_{10}$    \\\hline
\end{tabular}
\end{minipage}
\begin{minipage}[b]{0.24\linewidth}
\centering
\vspace{0.2cm}
\begin{tabular}{|c|p{0.19cm}|p{0.19cm}|p{0.19cm}|p{0.19cm}|}\hline          $x_A=3$ & 0 & 1 & 2 & 3 \\\hline \hline
        0   & $\mathcal{L}_{16}$  &  $\pmb{\mathcal{L}}_{\pmb{7}}$   &$\pmb{\mathcal{L}}_{\pmb{6}}$   &  $\mathcal{L}_{15}$    \\\hline
        1   & $\mathcal{L}_{14}$  &  $\mathcal{L}_{13}$   &  $\mathcal{L}_{12}$  & $\mathcal{L}_{11}$   \\\hline
        2   & $\pmb{\mathcal{L}}_{\pmb{4}}$    &  $\pmb{\mathcal{L}}_{\pmb{3}}$     &   $\mathcal{L}_{10}$   &  $\mathcal{L}_{9}$     \\\hline
        3   & $\pmb{\mathcal{L}}_{\pmb{2}}$     &$\pmb{\mathcal{L}}_{\pmb{1}}$  &  $\pmb{\mathcal{L}}_{\pmb{5}}$     &  $\mathcal{L}_{8}$     \\\hline
\end{tabular}
\end{minipage}
\vspace{0.2cm}
%\end{tabular}
\caption{Clustering for Example 2}
\label{fig:ex2fig2}}
\end{figure*}
\end{example}

For each of the 112 possibilities of singular fade subspaces of \textit{Case 3} of the previous section, a clustering of size between 16 to 23 can be achieved by first constraining the array representing the relay map in order to remove the singular fade state and then completing the constrained array using the provided algorithm. \\

\section{SIMULATION RESULTS}

The proposed scheme is based on the removal of singular fade states for the three-way relaying sccenario. A minimum cluster distance greater than zero is ensured for all the fade states, excluding a subset of singular fade states, which are shown to be non-removable. It is attempted to ensure that in the given scenario, the number of clusters in the clustering, which is the same as the size of the signal set used during the BC phase is minimized. Simulation results presented in this section identify some cases where the proposed scheme outperforms the naive approach that uses the same map for all fade states and vice verse All the simulation results shown in this section are for the case when the end nodes use 4-PSK signal set. The simulation results for the end to end BER as a function of SNR are presented in this section for different fading scenarios.\\

Consider the case when $H_A, H_B, H_C, H'_A, H'_B$ and $H'_C$ are distributed according to Rician distribution and channel variances equal to 0 dB. The SNR vs BER curve for this case, for a frame length of 256 bits. The plots for the cases with a Rician Factors of 5 dB, 10 dB, 15 dB and 20 dB are as shown in Fig. 13, Fig. 14, Fig. 15 and Fig. 16 respectively. The figures show the SNR vs bit-error-rate curves for the following schemes: adaptive clustering given in the paper, and non-adaptive clustering. For non-adaptive clustering, the relay node uses the same map given by Figure \ref{fig:xor} for all the channel realisations, we refer to this map as the non-adaptive map. It can be seen from Fig. \ref{fig:plot_bc_rician} that the schemes based on the adaptive clustering relaying perform better than the schemes based on non-adaptive clustering at low SNR, since adaptive clustering removes 112 singular fade states. \\ 

%The diversity gain in the case of adaptive network coding is 1, and adaptive network coding provides a gain of 2 dB over non-adaptive network coding.
It can be seen from the simulation results presented that the non-adaptive network coding performs better than the proposed scheme above a certain SNR when there is a dominant line of sight component, as in the case of Rician fading scenario. The reason for this is as follows: The end to end SNR vs BER as well as the throughput performance depend on the performance during the MA phase as well as the BC phase. As the line of sight component becomes more and more predominant, the performance during the BC phase gets better and better, but the effect of multiple access interference which occurs during the MA phase remains the same. Hence, for the cases when line of sight component is predominant, the performance degradation due to the MA interference predominates over the degradation occurring during the BC phase. The case of non-adaptive network coding, the relay utilizes a constellation of least possible size, whereas in adaptive network coding, the relay attempts at removing singular fade states, thereby optimizing the performance during MA phase to the fullest extent possible, at the cost of degraded performance during the BC phase. Hence, the non-adaptive network coding performs better than the proposed scheme at high SNR.\\

\begin{figure*}[ht]
\centering
{\tiny
%\begin{tabular}{c}
\vspace{0.2cm}
\begin{minipage}[b]{0.24\linewidth}
\centering
\begin{tabular}{|c|p{0.19cm}|p{0.19cm}|p{0.19cm}|p{0.19cm}|}\hline          $x_A=0$ & 0 & 1 & 2 & 3 \\\hline \hline
        0   &$\mathcal{L}_{1}$   &$\mathcal{L}_{2}$   &$\mathcal{L}_{3}$   &$\mathcal{L}_{4}$   \\\hline
        1   &$\mathcal{L}_{5}$   &$\mathcal{L}_{6} $  &$\mathcal{L}_{7}$   &$\mathcal{L}_{8}$  \\\hline
        2   &$\mathcal{L}_{9} $  &$\mathcal{L}_{10}  $ &$\mathcal{L}_{11}$   &$\mathcal{L}_{12}$   \\\hline
        3   &$\mathcal{L}_{13}$   &$\mathcal{L}_{14}$  &$\mathcal{L}_{15}$    &$\mathcal{L}_{16}$  \\\hline
\end{tabular}
\end{minipage} 
\begin{minipage}[b]{0.24\linewidth}
\centering
\begin{tabular}{|c|p{0.19cm}|p{0.19cm}|p{0.19cm}|p{0.19cm}|}\hline          $x_A=1$ & 0 & 1 & 2 & 3 \\\hline \hline
        0   &$\mathcal{L}_{6}$   &$\mathcal{L}_{5}$   &$\mathcal{L}_{8}$   &$\mathcal{L}_{7}$   \\\hline
        1   &$\mathcal{L}_{2}$   &$\mathcal{L}_{1} $  &$\mathcal{L}_{4}$   &$\mathcal{L}_{3}$  \\\hline
        2   &$\mathcal{L}_{14} $  &$\mathcal{L}_{13}  $ &$\mathcal{L}_{16}$   &$\mathcal{L}_{15}$   \\\hline
        3   &$\mathcal{L}_{10}$   &$\mathcal{L}_{9}$  &$\mathcal{L}_{12}$    &$\mathcal{L}_{11}$  \\\hline
\end{tabular}
\end{minipage} 
\begin{minipage}[b]{0.24\linewidth}
\centering
\begin{tabular}{|c|p{0.19cm}|p{0.19cm}|p{0.19cm}|p{0.19cm}|}\hline          $x_A=2$ & 0 & 1 & 2 & 3 \\\hline \hline
        0   &$\mathcal{L}_{11}$   &$\mathcal{L}_{12}$   &$\mathcal{L}_{9}$   &$\mathcal{L}_{10}$   \\\hline
        1   &$\mathcal{L}_{15}$   &$\mathcal{L}_{16} $  &$\mathcal{L}_{13}$   &$\mathcal{L}_{14}$  \\\hline
        2   &$\mathcal{L}_{3} $  &$\mathcal{L}_{4}  $ &$\mathcal{L}_{1}$   &$\mathcal{L}_{2}$   \\\hline
        3   &$\mathcal{L}_{7}$   &$\mathcal{L}_{8}$  &$\mathcal{L}_{5}$    &$\mathcal{L}_{6}$  \\\hline
\end{tabular}
\end{minipage} 
\hspace{0.2cm}
\begin{minipage}[b]{0.24\linewidth}
\vspace{0.2 cm}
\centering
\begin{tabular}{|c|p{0.19cm}|p{0.19cm}|p{0.19cm}|p{0.19cm}|}\hline          $x_A=3$ & 0 & 1 & 2 & 3 \\\hline \hline
        0   &$\mathcal{L}_{16}$   &$\mathcal{L}_{15}$   &$\mathcal{L}_{14}$   &$\mathcal{L}_{13}$   \\\hline
        1   &$\mathcal{L}_{12}$   &$\mathcal{L}_{11} $  &$\mathcal{L}_{10}$   &$\mathcal{L}_{9}$  \\\hline
        2   &$\mathcal{L}_{8} $  &$\mathcal{L}_{7}  $ &$\mathcal{L}_{6}$   &$\mathcal{L}_{5}$   \\\hline
        3   &$\mathcal{L}_{4}$   &$\mathcal{L}_{3}$  &$\mathcal{L}_{2}$    &$\mathcal{L}_{1}$  \\\hline
\end{tabular}
\end{minipage}
\vspace{0.2cm}}
%\end{tabular}}
\caption{Non-Adaptive map}
\label{fig:xor}
\end{figure*}

\begin{figure}[tp]
\centering
\includegraphics[totalheight=3in,width=3.8in]{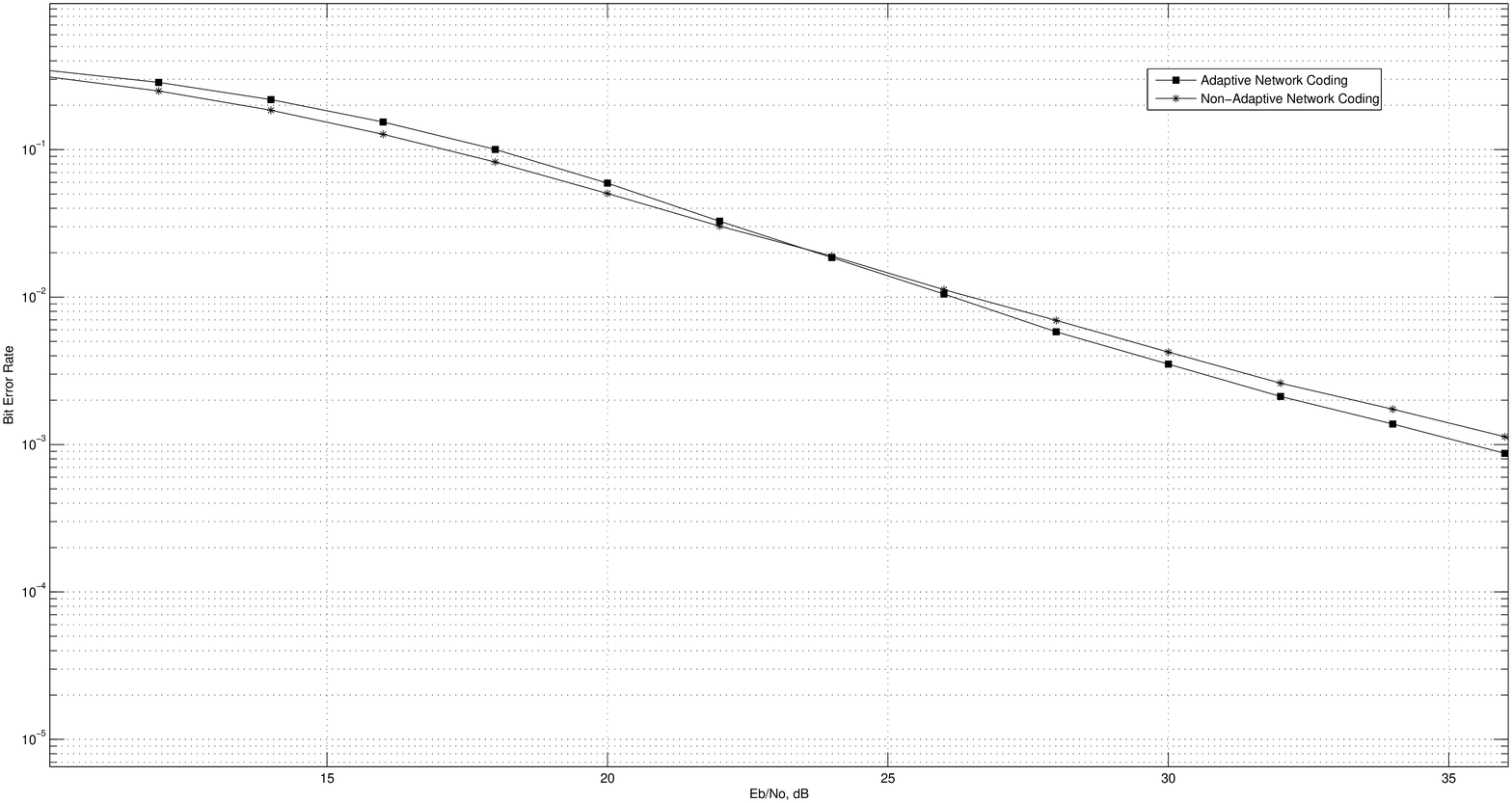}
\caption{SNR vs ber curves for different schemes for 4-PSK signal set when the Rician Factors is 5 dB}	
\label{fig:plot_bc_rician5}	
\end{figure}
\begin{figure}[tp]
\centering
\includegraphics[totalheight=3in,width=3.8in]{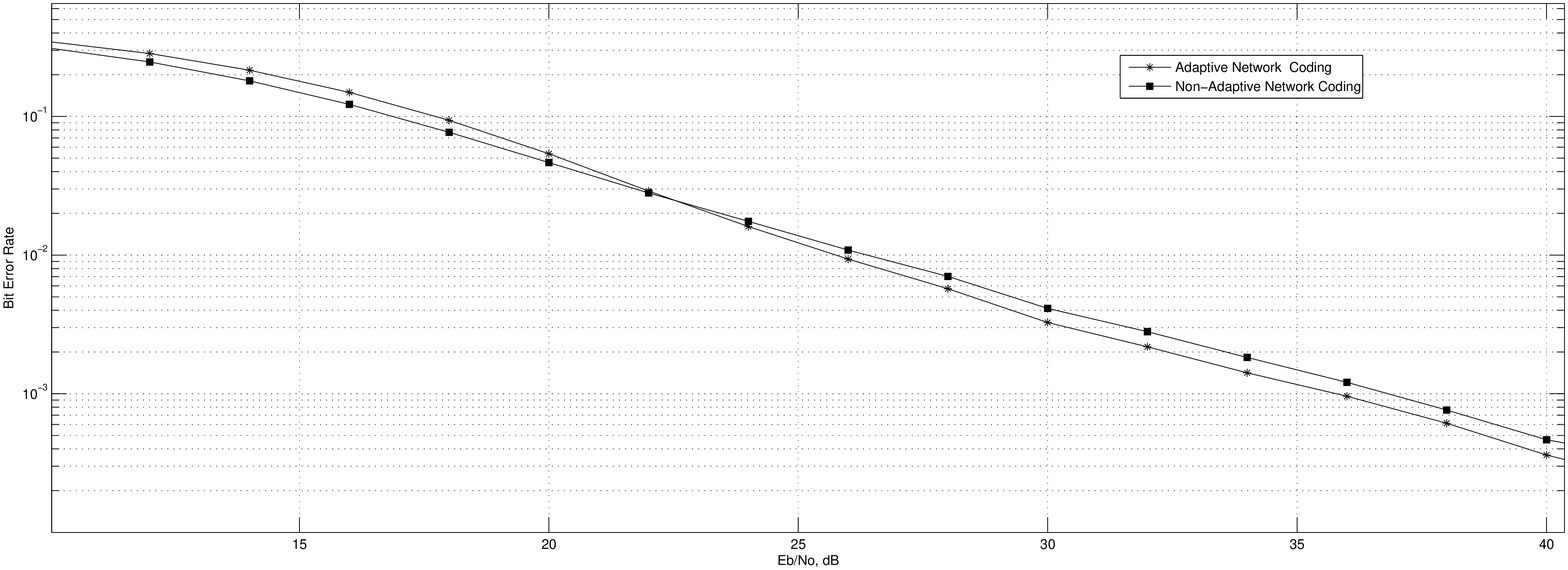}
\caption{SNR vs ber curves for different schemes for 4-PSK signal set when the Rician Factors is 10 dB}	
\label{fig:plot_bc_rician}	
\end{figure}
\begin{figure}[tp]
\centering
\includegraphics[totalheight=3in,width=3.8in]{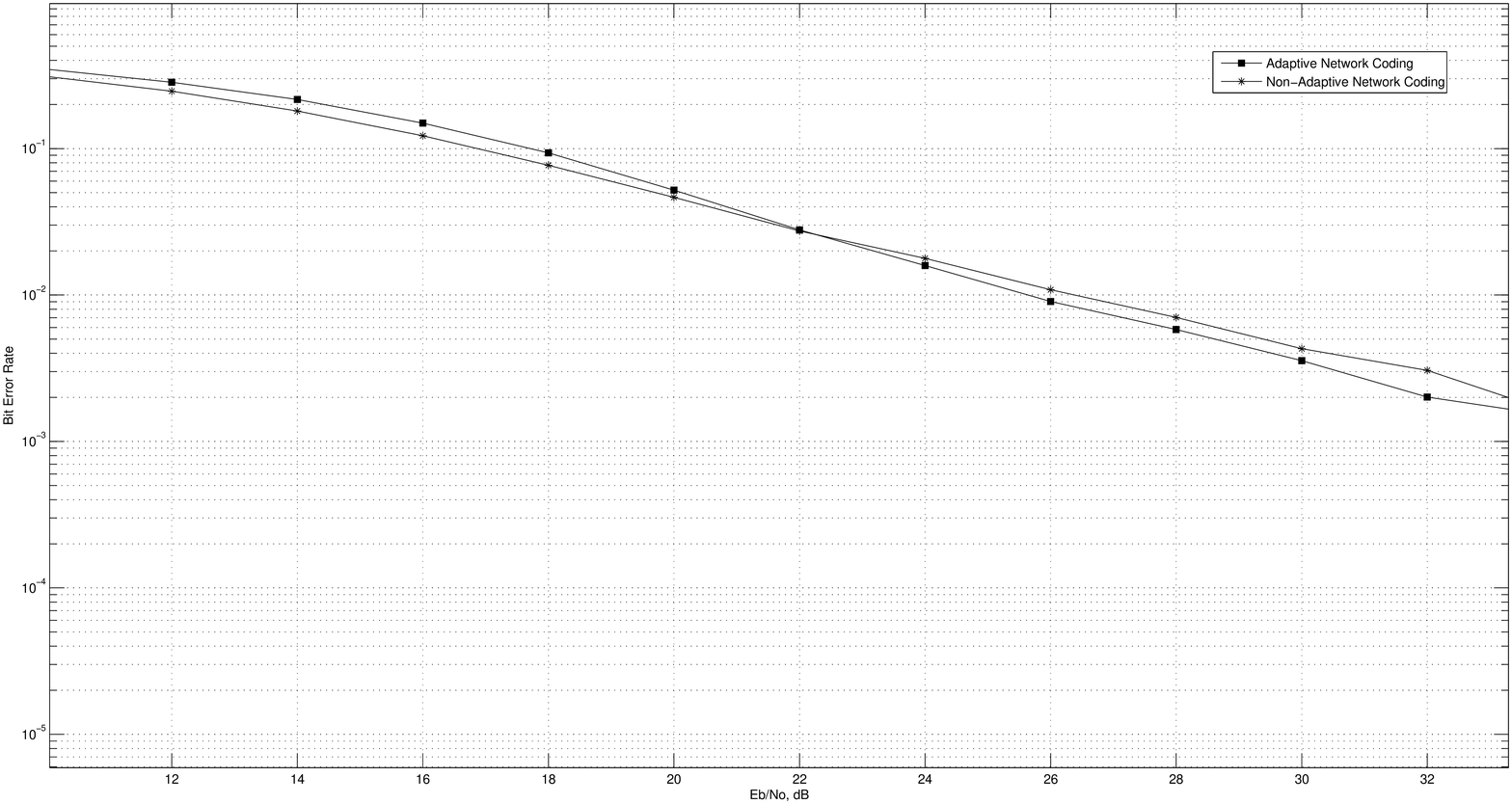}
\caption{SNR vs ber curves for different schemes for 4-PSK signal set when the Rician Factors is 15 dB}	
\label{fig:plot_bc_rician15}	
\end{figure}
\begin{figure}[tp]
\centering
\includegraphics[totalheight=3in,width=3.8in]{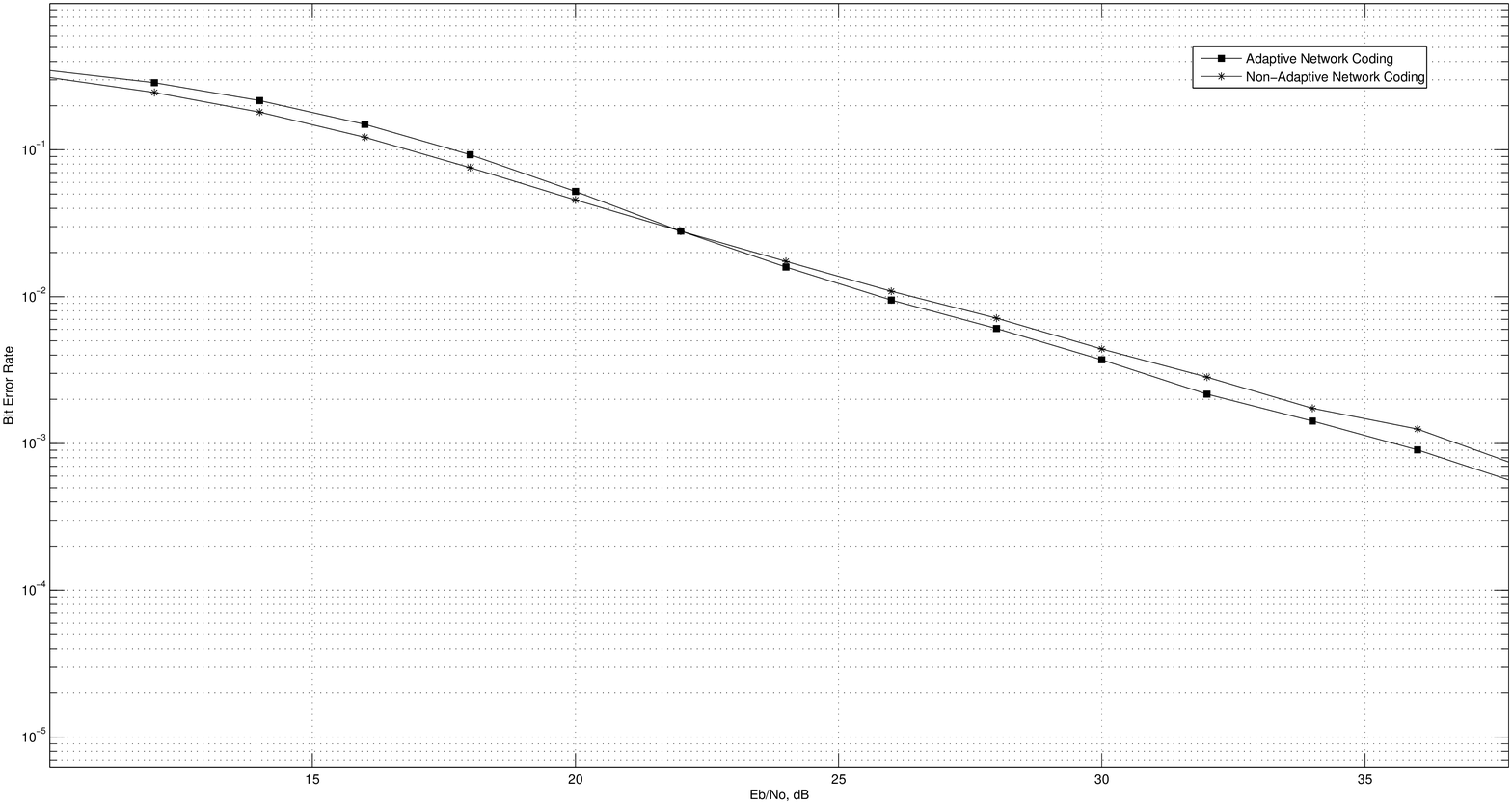}
\caption{SNR vs ber curves for different schemes for 4-PSK signal set when the Rician Factors is 20 dB}	
\label{fig:plot_bc_rician20}	
\end{figure}
%
%\begin{figure}[htbp]
%\centering
%\includegraphics[totalheight=3.25in,width=3.75in]{tput_curves_8psk.eps}
%\caption{SNR vs throughput curves for different schemes for 8-PSK signal set}	
%\label{fig:tput_curves_8psk}	
%\end{figure}

%The maximum throughput achieved by the ACF relaying schemes is 8/3 bits/s/Hz, whereas it is 2 bits/s/Hz for the 2-stage relaying schemes. Also, it can be seen from Fig. \ref{fig:tput_curves_4psk}, the scheme based on the removal of singular fade states using Latin Squares performs better than the CP based scheme. The reason for this is that the maximum cardinality of the signal set used during the BC phase is 25 for the CP based scheme whereas it is 20 for the Latin Square based scheme.\\

%Fig. \ref{fig:tput_curves_8psk} shows the SNR vs throughput curves for the different schemes, for the case when 8-PSK signal set is used by the nodes A and B. At high SNR, the CP based scheme provides a larger throughput than the 2-stage relaying schemes. The maximum throughput achieved by the CP based scheme is 4 bits/s/Hz, whereas it is 3 bits/s/Hz for the 2-stage relaying schemes.\\

\section{Conclusion}
Our paper deals with the three-way wireless relaying scenario, assuming that the three nodes operate in half-duplex mode and that they transmit points from the same 4-PSK constellation. It is shown that it is possible for information exchange to take place using just two channels uses, unlike the other work done for the case, to the best of our knowledge. The Relay node clusters the $4^{3}$ possible transmitted tuples $\left(x_{A},x_{B},x_{C}\right)$ into various clusters such that \textit{the exclusive law} is satisfied. This necessary requirement of satisfying the exclusive law is shown to be the same as the clustering being represented by a Latin Cube of second order. Using the proposed schemes, not only is the exchange of information between the three nodes made possible using three channel uses, the size of the resulting constellation used by the relay node R in the BC phase is reduced from $4^{3}$ to lie between 16 to 23. Note that we do not claim that the size of the clustering utilizing modified clustering is the best that can be achieved, since our method of filling the Latin Cube of Second Order of side 4 may not be the most optimal process of doing so, and it might be possible to fill the array with less than 23 symbols. \\

\section*{Acknowledgement}
 This work was supported partly by the DRDO-IISc program on Advanced Research in Mathematical Engineering through a research grant as well as the INAE Chair Professorship grant to B. S. Rajan.\\

\end{document}